\documentclass[11pt]{article}
\usepackage[margin=1in]{geometry}
\usepackage{lmodern}
\usepackage[T1]{fontenc}
\usepackage[utf8]{inputenc}
\usepackage[tbtags]{amsmath}
\allowdisplaybreaks
\usepackage{amssymb,amsthm,mathtools}
\usepackage{hyperref}
\usepackage[svgnames]{xcolor}
\definecolor{links}{RGB}{11, 85, 255}
\definecolor{cites}{RGB}{0, 200, 0}

\definecolor{urls}{RGB}{255, 116, 0}
\hypersetup{colorlinks={true},linkcolor={links},citecolor=[named]{cites},urlcolor={urls}}
\usepackage{doi}
\usepackage{tikz} 
\usepackage{pgfplots}
\pgfplotsset{compat=1.14}
\usepgfplotslibrary{fillbetween}
\usepackage{hyphenat}
\usepackage[font=footnotesize]{caption}
\usepackage{subcaption}
\usepackage{appendix}
\usepackage{todonotes}

\usepackage{booktabs} 
\usepackage[ruled]{algorithm2e} 

\SetAlFnt{\small}
\SetAlCapFnt{\small}
\SetAlCapNameFnt{\small}
\SetAlCapHSkip{0pt}
\IncMargin{-\parindent}

\usepackage[capitalise,nameinlink]{cleveref}
\usepackage{float}
\usepackage{multicol}
\usepackage{enumitem}
\usepackage{tcolorbox}
\usepackage{physics}
\usepackage{multirow}
\usepackage{natbib}
\usepackage{cleveref}

\usepackage{thm-restate}

\usepackage {tikz}
\usetikzlibrary {positioning}
\usetikzlibrary{shapes}
\definecolor {processblue}{cmyk}{0.96,0,0,0}
\usepackage{pgfplots}

\usepackage{multirow,array, multicol}

 \hypersetup{
           breaklinks=true,   
           colorlinks=true,   
        }
\newcommand{\E}{\textbf{E}}
\newcommand{\ALG}{\textsf{ALG}}
\newcommand{\OPT}{\textsf{OPT}}
\newcommand{\RoE}{\textsf{RoE}}
\newcommand{\SRoE}{\textsf{S-RoE}}
\newcommand{\AALG}{\ALG_{\RoE}}
\newcommand{\SALG}{\textsf{S-}\ALG_{\RoE}}
\newcommand{\ALGEOR}{\ALG_{\EoR}}
\newcommand{\PbM}{\textsf{PbM}}
\newcommand{\EoR}{\textsf{EoR}}
\newcommand{\SEoR}{\textsf{S-EoR}}
\newcommand{\EoIR}{\textsf{EoIR}}

\newcommand{\bone}{\boldsymbol{1}}
\newcommand{\bx}{\boldsymbol{x}}

\newcommand{\bv}{\boldsymbol{v}}
\newcommand{\bu}{\boldsymbol{u}}
\newcommand{\bw}{\boldsymbol{w}}

\newcommand{\cE}{\mathcal{E}}

\newcommand{\cF}{\mathcal{F}}

\newcommand{\cX}{\mathcal{X}}
\newcommand{\cU}{\mathcal{U}}

\newcommand{\given}{\ \bigg| \ }

\newtheorem{theorem}{Theorem}
\newtheorem*{theorem*}{Theorem}
\newtheorem{lemma}{Lemma}
\newtheorem{proposition}{Proposition}
\newtheorem{corollary}{Corollary}

\theoremstyle{definition}
\newtheorem{definition}{Definition}

\newtheorem{claim}{Claim}
\newtheorem{remark}{Remark}
\newtheorem{example}{Example}
\newtheorem{observation}{Observation}

\usepackage{graphicx}

\begin{document}
\title{Prophet Inequalities via the Expected Competitive Ratio}

\date{}
\author{Tomer Ezra\thanks{Simons Laufer Mathematical Sciences Institute, Email:  \texttt{tomer.ezra@gmail.com}} 
\and Stefano Leonardi\thanks{Sapienza University of Rome, Email:  \texttt{leonardi@diag.uniroma1.it}} 
\and Rebecca Reiffenh\"auser\thanks{University of Amsterdam, Email:  \texttt{r.e.m.reiffenhauser@uva.nl}}
\and Matteo Russo\thanks{Sapienza University of Rome, Email:  \texttt{mrusso@diag.uniroma1.it}}
\and Alexandros Tsigonias-Dimitriadis\thanks{European Central Bank, Email:  \texttt{alexandrostsigdim@gmail.com}}
}

\maketitle              
\begin{abstract}
    We consider prophet inequalities under general downward-closed constraints. In a prophet inequality problem, a decision-maker sees a series of online elements with values, and needs to decide immediately and irrevocably whether or not to select each element upon its arrival, subject to an underlying feasibility constraint.  
Traditionally, the decision-maker's expected performance has been compared to the expected performance of the \emph{prophet}, i.e., the expected offline optimum. We refer to this measure as the \textit{Ratio of Expectations} (or, in short, \RoE). However, a major limitation of the $\RoE$ measure is that it only gives a guarantee against what the optimum would be on average, while, in theory, algorithms still might perform poorly compared to the realized ex-post optimal value. Hence, we study alternative performance measures. In particular, we suggest the \textit{Expected Ratio} (or, in short, $\EoR$), which is the expectation of the ratio between the value of the algorithm and the value of the prophet. This measure yields desirable guarantees, e.g., a constant $\EoR$ implies achieving a constant fraction of the ex-post offline optimum with constant probability. Moreover, in the single-choice setting, we show that the $\EoR$ is equivalent (in the worst case) to the probability of selecting the maximum, a well-studied measure in the literature. This is no longer the case for combinatorial constraints (beyond single-choice), which is the main focus of this paper. Our main goal is to understand the relation between $\RoE$ and $\EoR$ in combinatorial settings. Specifically, we establish two reductions: for every feasibility constraint, the $\RoE$ and the $\EoR$ are at most a constant factor apart. Additionally, we show that the $\EoR$ is a stronger benchmark than the $\RoE$ in that for every instance (feasibility constraint and product distribution) the $\RoE$ is at least a constant fraction of the $\EoR$, but not vice versa. Both these reductions imply a wealth of $\EoR$ results in multiple settings where $\RoE$ results are known. 
\end{abstract}

\section{Introduction}\label{sec:intro}

Prophet Inequalities are one of optimal stopping theory's most prominent problem classes. In the classic prophet inequality, a decision-maker must select an element $e$ from an online sequence of elements $E$ immediately and irrevocably. The sequence is revealed one by one in an online fashion, and the decision-maker wants to maximize the weight of the chosen element, where each element's weight is drawn from some distribution $D_e$. The decision-maker knows the distributions and is compared to a \emph{prophet}, who knows all the realizations of the weights in advance. A classic result of \citet{krengel,krengel2}, and \citet{samuel-cahn} asserts that the decision-maker can extract at least half of the prophet's expected reward and that this result is tight.

A vast body of research has studied the classic prophet inequality and its variants, where the objective function is to maximize the ratio between what the algorithm gets in expectation and the expected weight of the ex-post optimum. We use the shorthand $\RoE$ to signify this \emph{ratio of expectations}. However, this benchmark has shortcomings for many applications of prophet inequalities. Oftentimes, the decision-maker is not only concerned about the expected value, but she also wants to have some guarantees with respect to the ex-post outcome. The concept of risk aversion has been defined in various ways in the literature: a common underlying principle is that the involved parties often want to avoid the possibility of \emph{extremely bad} outcomes.

As our first example shows, such risk-averse decision-makers might prefer to select a box with a deterministic weight of $1$, even though the second box's expected weight is slightly larger. This is because the weight of the second box has a high probability of having a value of $0$, and it is much riskier to choose this option for just a marginal improvement in the expected utility.
\begin{example}\label{ex:ex1}
Consider a setting with two boxes. The first box's weight $w_1$  is deterministically $1$, and the second box's  weight $w_2$ is $0$ with probability $1-\varepsilon$ and $\frac{1+2\varepsilon}{\varepsilon}$ with probability $\varepsilon$, for $\varepsilon \in (0, 1]$.
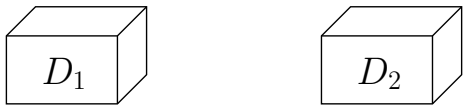
\begin{figure}[H]
    \centering
    \scalebox{0.7}{\tikzset{every picture/.style={line width=0.75pt}} 
\begin{tikzpicture}[x=0.75pt,y=0.75pt,yscale=-1,xscale=1]
    \draw   (164,111) -- (185,90) -- (265,90) -- (265,139) -- (244,160) -- (164,160) -- cycle ; \draw   (265,90) -- (244,111) -- (164,111) ; \draw   (244,111) -- (244,160) ;
    \draw   (391,111) -- (412,90) -- (492,90) -- (492,139) -- (471,160) -- (391,160) -- cycle ; \draw   (492,90) -- (471,111) -- (391,111) ; \draw   (471,111) -- (471,160) ;
    \draw (188,125) node [anchor=north west][inner sep=0.75pt]  [font=\huge]  {$D_1$};
    \draw (415,125) node [anchor=north west][inner sep=0.75pt]  [font=\huge]  {$D_2$};
\end{tikzpicture}}
    \caption{Two boxes: $w_1 = 1$, $w_2 \sim D_2 = \begin{cases} 0, \text{ w.p. } 1 - \varepsilon \\ \frac{1+2\varepsilon}{\varepsilon}, \text{ w.p. } \varepsilon \end{cases}$.}
    \label{fig:ex1}
\end{figure}
The decision-maker's expected utility would be $1$ if she selects the first box, and is $1+2\varepsilon$ if she selects the second box. While picking the second box  maximizes the $\RoE$, this is a much riskier choice that brings only a negligible improvement.
\end{example}

Since maximizing the $\RoE$ does not capture the phenomenon of risk aversion, we would like to define a benchmark that does. A first suggestion is the \emph{probability of selecting the maximum} ($\PbM$), introduced by \citet{GM66} for the case of i.i.d. valued elements, for which the $\PbM = 0.5801$. In the non-i.i.d. case, in worst-case order, \citet{esfandiari} show a tight bound on $\PbM$ of $1/e$.

A decision-maker that maximizes the $\PbM$ selects the first box in \Cref{ex:ex1}, and thus picks the maximum with probability close to $1$. Another different approach from the $\PbM$ is the \emph{expected ratio}, $\EoR$, between the algorithm's weight and the weight of the ex-post optimum (originally suggested in \cite{ScharbrodtSS06} for other domains). In \Cref{app:single}, we establish that $\PbM$ and $\EoR$ are essentially identical measures of performance in single-choice settings. This no longer holds for richer variants of prophet inequalities beyond single-choice, which is the main focus of our paper.\\ 

We study the natural extension of classic prophet inequalities, termed prophet inequalities with combinatorial constraints \citep{rubinstein-matr}, where the decision-maker is allowed to select more than a single element according to a predefined (downward-closed) feasibility constraint. These types of constraints capture the idea that if a given set is feasible, so are all its subsets: examples include knapsack, matchings, and general matroids, as well as their intersection. In the next example, one can observe that, for any online algorithm, the probability of selecting the maximum is exponentially small in the number of elements; put another way, the probability of selecting the exact optimum offline is negligible and, thus, the guarantees of the $\PbM$ measure do not extend to combinatorial prophet inequalities.
\begin{example}\label{ex:ex2}
Consider a setting with $n$ pairs of boxes, such that, for each pair $i$, one box has weight $w_{1,i} = 1$ deterministically, and the second has weight $w_{2,i}$, equal to $0$ with probability $1/2$ and to $2$ with probability $1/2$. The (downward-closed) feasibility constraint is that at most one box from each pair can be selected (a partition matroid), and the decision-maker gets the sum over the selected set. 
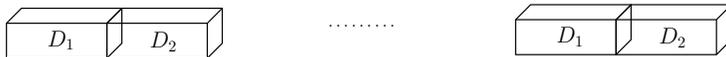
\begin{figure}[H]
    \centering
    \scalebox{0.6}{\tikzset{every picture/.style={line width=0.75pt}} 

\begin{tikzpicture}[x=0.75pt,y=0.75pt,yscale=-1,xscale=1]

\draw   (20,89.6) -- (32.6,77) -- (117,77) -- (117,106.4) -- (104.4,119) -- (20,119) -- cycle ; \draw   (117,77) -- (104.4,89.6) -- (20,89.6) ; \draw   (104.4,89.6) -- (104.4,119) ;
\draw   (104.4,89.6) -- (117,77) -- (201.4,77) -- (201.4,106.4) -- (188.8,119) -- (104.4,119) -- cycle ; \draw   (201.4,77) -- (188.8,89.6) -- (104.4,89.6) ; \draw   (188.8,89.6) -- (188.8,119) ;
\draw   (448,86.6) -- (460.6,74) -- (545,74) -- (545,103.4) -- (532.4,116) -- (448,116) -- cycle ; \draw   (545,74) -- (532.4,86.6) -- (448,86.6) ; \draw   (532.4,86.6) -- (532.4,116) ;
\draw   (532.4,86.6) -- (545,74) -- (629.4,74) -- (629.4,103.4) -- (616.8,116) -- (532.4,116) -- cycle ; \draw   (629.4,74) -- (616.8,86.6) -- (532.4,86.6) ; \draw   (616.8,86.6) -- (616.8,116) ;

\draw (289,89.4) node [anchor=north west][inner sep=0.75pt]    {$\dotsc \dotsc \dotsc $};
\draw (53,94.4) node [anchor=north west][inner sep=0.75pt]  [font=\Large]  {$D_{1}$};
\draw (139,95.4) node [anchor=north west][inner sep=0.75pt]  [font=\Large]  {$D_{2}$};
\draw (481,91.4) node [anchor=north west][inner sep=0.75pt]  [font=\Large]  {$D_{1}$};
\draw (567,92.4) node [anchor=north west][inner sep=0.75pt]  [font=\Large]  {$D_{2}$};

\end{tikzpicture}}
    \caption{$n$ pairs of boxes: for each pair, $w_{1,i} = 1$, $w_{2,i} \sim D_2 = \begin{cases} 0, \text{ w.p. } 1/2 \\ 2, \text{ w.p. } 1/2 \end{cases}$.}
    \label{fig:ex2}
\end{figure}
The probability of selecting the maximum is the probability that the algorithm chooses the larger realized value for each pair of boxes. An online algorithm cannot select the maximum of each pair with probability greater than $1/2$. Since all realizations are independent, we have an upper bound of $\PbM = 1/2^n$ for every online algorithm (see Claim \ref{clm:negl-pbm}).
This motivates choosing a different measure of performance in combinatorial settings. In particular, for this example, the algorithm that always selects the first box for each pair guarantees good expected ratio. Indeed, by Jensen's Inequality (see~\Cref{app:omitted}), we have
\begin{align*}
    \E\left[\frac{\text{\ALG}}{\text{\OPT}}\right] &\geq \frac{n}{\E\left[\sum_{i \in [n]}\max\{w_{1,i}, w_{2,i}\}\right]} = \frac{n}{\frac{3}{2}n} = \frac{2}{3}.
\end{align*}
Note that this algorithm also guarantees $\ALG \geq \left(\frac{2}{3} - \varepsilon\right) \cdot \OPT$ with high probability, which is a type of guarantee a risk-averse decision-maker would desire.
\end{example}

Combining that $\EoR$ and $\PbM$ are equivalent for single-choice settings and that $\PbM$ is unachievable beyond single-choice, we believe that the $\EoR$ is the right alternative to the $\PbM$ in combinatorial settings. In \Cref{app:alternative}, we discuss other possible extensions of $\PbM$  and show their shortcomings in combinatorial settings.\\

It is important to note that there are some instances where optimizing $\EoR$ leads to bad guarantees for risk-averse decision-makers. This is the case for both the $\EoR$ and the $\RoE$ (see \Cref{ex:risk} in \Cref{app:risk}).
However, there are cases where having to average over many runs to obtain a good ratio does not meet the problem requirements.
Consider a platform or marketplace that faces a frequently repeated (e.g., daily) resource allocation problem. The task is to allocate limited resources to a stream of customers, subject to any underlying downward-closed constraint about what can be allocated and to whom. The platform wants to maximize an objective, such as social welfare or revenue. From the perspective of both the platform and the customers, it is often desirable to know that some ex-post guarantees will be satisfied. Specifically, the customers on a given day might want to know that if they have a high value for some subset of the resources, they will have a fair chance at getting it, and the platform also wants to ensure that on every instance, it will allocate a good fraction of the resources to the customers who value them highly. In such scenarios, designing a strategy that maximizes the $\EoR$, rather than the $\RoE$ or some other performance measure, will guarantee that.  

In \Cref{app:const-eor}, we provide a series of claims (Claim \ref{cl:imply}, Claim \ref{clm:one-over-e}, and Claim \ref{clm:two-thirds}), which establish that our definition of $\EoR$ (in contrast to the $\RoE$) guarantees the best-we-can-hope-for when minimizing the risk compared to the ex-post value. More specifically, Claim \ref{cl:imply} shows that a ``good'' (i.e., relatively high constant) $\EoR$ directly implies that we attain a constant fraction of the optimum with constant probability. Moreover, as shown in Claim \ref{clm:one-over-e} and Claim \ref{clm:two-thirds}, no (qualitatively) better bi-criteria approximation can be achieved; there exist simple feasibility constraint-distribution pairs for which either no constant approximation to the maximum is possible with high probability, or no near-optimal approximation to the maximum can be attained with constant probability. Therefore, settling for the ex-post guarantee of the $\EoR$ is best possible in combinatorial prophet inequalities. Moreover, instead of aiming directly for such bi-criteria results, our main goal in this paper is to suggest a natural alternative performance measure and present its properties and insights it provides. Thus, we believe that, apart from its simplicity, the two main reasons that make the $\EoR$ an interesting objective function are (1) that is the ``right'' generalization of $\PbM$ beyond the single-choice setting, and, (2) that it captures well one of the natural ways to think of risk-aversion in online decision-making.





We further investigate how the notions of the ratio of expectations and the expected ratio are connected to each other. As a first step, the following examples show why a constant $\RoE$ algorithm does not guarantee a constant $\EoR$, and vice versa.

First, consider \Cref{ex:ex1}. The canonical $1/2$-competitive (and tight) $\RoE$ algorithm for the single-choice problem is that of setting a single threshold $\tau := \E\left[\text{OPT}\right]/2 = \E\left[\max\left\{w_1, w_2\right\}\right]/2 = (1 - \varepsilon + 1 + 2\varepsilon)/2 = 1 + \varepsilon/2$, and accepting the first box whose weight exceeds $\tau$. We now analyze the performance of such algorithm, measured according to $\EoR$:
\begin{align*}
    \EoR := \E\left[\frac{\ALG}{\OPT}\right] = (1 - \varepsilon) \cdot \frac{0}{1} + \varepsilon \cdot \frac{(1+2\varepsilon)/\varepsilon}{(1+2\varepsilon)/\varepsilon} = \varepsilon,
\end{align*}
since the algorithm would only accept if the value is at least $1 + \varepsilon/2$, which only happens if the second box realization is $(1+2\varepsilon)/\varepsilon$. This algorithm has no $\EoR$ guarantee since $\varepsilon$ can be arbitrarily small.

Second, the next example shows that a constant $\EoR$ algorithm is not necessarily constant competitive in the $\RoE$ sense.

\begin{example}\label{ex:ex3}
Consider a setting with two boxes, one with a weight $w_1 = 1$ deterministically, and the second with a weight $w_2$, which is $\varepsilon^2$ with probability $1-\varepsilon$ and $1/\varepsilon^2$ with probability $\varepsilon$, for $\varepsilon \in (0, 1]$.
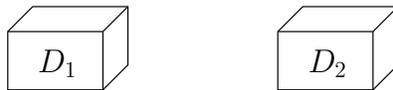
\begin{figure}[H]
    \centering
    \scalebox{0.6}{\tikzset{every picture/.style={line width=0.75pt}} 
\begin{tikzpicture}[x=0.75pt,y=0.75pt,yscale=-1,xscale=1]
    \draw   (164,111) -- (185,90) -- (265,90) -- (265,139) -- (244,160) -- (164,160) -- cycle ; \draw   (265,90) -- (244,111) -- (164,111) ; \draw   (244,111) -- (244,160) ;
    \draw   (391,111) -- (412,90) -- (492,90) -- (492,139) -- (471,160) -- (391,160) -- cycle ; \draw   (492,90) -- (471,111) -- (391,111) ; \draw   (471,111) -- (471,160) ;
    \draw (188,125) node [anchor=north west][inner sep=0.75pt]  [font=\huge]  {$D_1$};
    \draw (415,125) node [anchor=north west][inner sep=0.75pt]  [font=\huge]  {$D_2$};
\end{tikzpicture}}
    \caption{Two boxes: $w_1 = 1$, $w_2 \sim D_2 = \begin{cases} \varepsilon^2, \text{ w.p. } 1 - \varepsilon \\ 1/\varepsilon^2, \text{ w.p. } \varepsilon \end{cases}$.}
    \label{fig:ex3}
\end{figure}
The algorithm that always selects the first box achieves $\E\left[\frac{\ALG}{\OPT}\right] = (1 - \varepsilon) \cdot \frac{1}{1} + \varepsilon \cdot \frac{1}{1/\varepsilon^2} > 1 - \varepsilon$. On the contrary, $\E\left[\ALG\right] = 1$ and $\E\left[\OPT\right] = 1 - \varepsilon + \frac{1}{\varepsilon} > \frac{1}{\varepsilon}$. Thus,
\begin{align*}
    \RoE := \frac{\E\left[\ALG\right]}{\E\left[\OPT\right]} < \frac{1}{1/\varepsilon} = \varepsilon.
\end{align*}
As $\varepsilon$ can be arbitrarily small, this algorithm does not guarantee a constant $\RoE$.
\end{example}

The aforementioned examples demonstrate that algorithms exhibiting a constant guarantee for one performance measure (such as the optimal $\RoE$ algorithm, which gives a guarantee of $1/2$ for any instance) might fail miserably in some instances for the other performance measure. This motivates us to deeper understand whether, and for which settings, a good algorithm for $\RoE$ can be transformed to a good algorithm for $\EoR$, and vice versa. The \textbf{guiding question} of this paper is, therefore,
\begin{center}
    \emph{What is the relation between $\RoE$ and $\EoR$?}
\end{center}


\subsection{Our contributions} 

As a motivation for our study, in \Cref{app:single}, we show the equivalence between $\PbM$ and $\EoR$ in the single-choice setting. We present two proofs for this equivalence; one is an adaptation of the worst-case example for the $\PbM$ measure by \citet{esfandiari}. The second proof is based on the observation that for each product distribution, we can construct a new product distribution for which the $\EoR$ is arbitrarily close to the $\PbM$ of the original distribution.

Our main results establish two reductions between $\EoR$ and $\RoE$.
In particular, we show that for every downward-closed feasibility constraint, $\RoE$ and $\EoR$ are at most a multiplicative constant factor apart 
(see \Cref{sec:relation} and \Cref{sec:reduction2}). For the following informal statements of the three main results, we first introduce some basic notation. We use $\RoE(\cF)$ (similarly $\EoR(\cF)$) to denote the ratio between the performance of the \emph{best} algorithm against the offline optimum on the \emph{worst-case} distribution, given a family of feasibility constraints $\cF$. Note that, in principle, we expect the worst-case distributions to be different for the two measures. In the second and third statement, we use the stronger notions of $\RoE(\cF,D)$ and $\EoR(\cF,D)$. Here, the ratio expresses the guarantee of the best algorithm against the offline optimum on the worst \emph{constraint-distribution pair $(\cF, D)$}. This means that the input now consists not only of a family $\cF$, but also of a product (i.e., the distributions of the elements' weights are independent) distribution $D$. 

\begin{theorem*}[Equivalence between $\RoE$ and $\EoR$, \Cref{col:eorroe}]
For every downward-closed family of feasibility constraints $\cF$ it holds that
\begin{align*}
\frac{\RoE(\cF)}{\EoR(\cF)} \in \Theta(1).
\end{align*}
\end{theorem*}

In the next result, we show that the $\EoR$ is a stronger benchmark than the $\RoE$ in the sense that for every instance composed of a feasibility constraint and a product distribution the $\RoE$ is at least a constant of the $\EoR$.

\begin{theorem*}[$\EoR$ to $\RoE$ reduction, \Cref{alg-analysis-rev-new}]
For every downward-closed family of feasibility constraints $\cF$, and a product distribution $D$ it holds that
    \begin{align*}
        \RoE(\cF,D) \geq \frac{\EoR(\cF,D)}{68}.
    \end{align*}
\end{theorem*}
We complement this by showing that the parallel result 
cannot be achieved in the other direction (i.e., from \RoE ~to \EoR).
\begin{theorem*}[$\RoE$ to $\EoR$ impossibility, \Cref{cor:noreduction}]
For every $\varepsilon>0$, there exist a feasibility constraint $\cF$  and a product distribution $D$ in which $\EoR(\cF,D)  \leq \varepsilon $ and  $\RoE(\cF,D)  \in \Omega(1)$.
\end{theorem*}

\subsection{Our techniques}

A key ingredient of our proof is a distinction between cases where the contribution to the value of the prophet comes from a large number of boxes and cases where the contribution mainly comes from a small set of boxes. 
If we are in the latter case, one can just run a simple threshold strategy and have a good guarantee. 
Otherwise, we use our second key ingredient which is analyzing the structure of the offline optimum function (value of the prophet) in the event that the threshold algorithm does not have a good enough guarantee.
In particular, we show that, under such event, the normalized offline optimum function is \emph{self-bounding} (see \Cref{def:bounding}), and therefore well-concentrated \citep{blm}. 
Our \Cref{prop:bounding} and \Cref{prop:bounding-xos} generalize claims about the self-bounding property of normalized offline optima shown in \cite{Vondrak10,BlumCHPPV17} (see \Cref{app:xos} for further discussion).



\smallskip

\noindent \textbf{Feasibility-based reduction:} $\RoE(\cF)$ \textbf{vs.} $\EoR(\cF)$ \textbf{.} To prove that the $\EoR(\cF)$ is at least a constant times $\RoE(\cF)$, we first calculate the threshold for which the maximal value exceeds it with probability of half.
We use this to perform a \emph{tail-core} split: Intuitively, if the expected offline optimum is not too large (i.e., close to the threshold), then our algorithm tries to catch a ``superstar'' (i.e., the first element with a weight above the threshold). To simplify our analysis, we count only the gain in the cases where exactly one such element is realized (the tail event). This happens with constant probability, and since the expected offline optimum is relatively small, we always get a good fraction of it by picking this unique element.

When instead, the expected offline optimum is large (i.e., far from the threshold), we run a constant competitive $\RoE$ algorithm in a black-box fashion. As already pointed out, an algorithm with constant $\RoE$ does not necessarily achieve any guarantee for $\EoR$. We overcome this obstacle through our case distinction and the self-bounding properties we show for the optimum. In particular, we upper bound the value of the offline optimum with high enough probability and lower bound the $\RoE$ algorithm expected value conditioned on the optimum not being too large. 

An immediate corollary of our result 
and \cite{rubinstein-matr} is that, for downward-closed feasibility constraints, the $\EoR$ is in $\Omega\left(\frac{1}{\log^2 n}\right)$. However, we prove a much stronger result: For every specific feasibility constraint, the $\EoR$ is a constant away from the $\RoE$ (which implies trivially the former assertion). In particular, if for some feasibility constraint the $\RoE$ is $\Omega\left(\frac{1}{\log\log n}\right)$, then our result shows that the $\EoR$ is approximately the same (and not that it is just bounded by $\Omega\left(\frac{1}{\log^2 n}\right)$). \smallskip 

\noindent \textbf{Instance-based reduction:} $\RoE(\cF, D)$ \textbf{vs.} $\EoR(\cF, D)$ \textbf{.} For the other direction, we show a stronger result, in that for every instance (a feasibility constraint, and a product distribution) the $\RoE$ is at least a constant fraction of the $\EoR$.

To achieve this result, we show that either the original $\EoR$ algorithm achieves up to a constant the same $\RoE$ guarantee, or that the simple threshold algorithm that achieves half of the expectation of the maximal element, guarantees a good $\RoE$.

We remark that both our reductions are constructive: We could interpret them as using a $\RoE$ algorithm black-box to design an $\EoR$ one, and vice versa. 
Further extensions and implications (XOS objectives, unknown prior, different assumptions on the arrival order, etc.) are discussed in \Cref{sec:discussion}. 

\subsection{Related Work}\label{relwork}

For early work on prophet inequalities, starting from the classic model and some of its most important variants, we refer the reader to the comprehensive survey of \citet{hill-kertz}. The topic of prophet inequalities has recently regained strong interest, primarily among researchers in theoretical computer science, due to its connections to (algorithmic) mechanism design and, in particular, posted price mechanisms \citep{hajiaghayi,chawla2,weinberg,correa-pricing}. The surveys of \citet{correa-survey} and \citet{lucier-survey} provide detailed overviews of recent results in prophet inequalities and their connections to mechanism design, respectively.

This recent surge of interest has given rise to a  stream of work, extending the classic prophet inequality to more general objective functions beyond single-choice (including submodular \citep{chekuri}, XOS \citep{feldman}, and monotone subadditive functions \citep{rubinstein-singla}), different assumptions on the arrival order, and rich combinatorial feasibility constraints. Among the latter, some notable results include $k$-uniform matroid \citep{hajiaghayi,alaei,jiang}, matching \citep{feldman,ezra,alaei2,gravin}, general matroid or knapsack \citep{chawla2,weinberg,feldman-ocrs,ehsani,dutting2}, and polymatroid constraints \citep{dutting}. Among the most general environments considered (in which non-trivial positive results can be achieved) are arbitrary downward-closed feasibility constraints \citep{rubinstein-matr,rubinstein-singla}. Combined with our framework, these results immediately give corresponding (lower and upper) bounds on the $\EoR$. Note that \citep{rubinstein-matr} also considers non-downward-closed feasibility constraints, but shows that it is impossible to achieve an $\RoE$ larger than $O(1/n)$.

One of our goals in this paper is to go beyond the traditional measure of performance in prophet inequalities, i.e., the ratio of expectations, and understand how natural alternatives perform in a wide range of scenarios. While the $\EoR$ measure has not been studied before in the context of prophet inequalities, \citet{GargGLS08} considered it for Bayesian cost minimization problems, such as the Online Stochastic Steiner tree problem, where they show an upper bound of $O(\log\log n)$ on the gap between $\RoE$ and $\EoR$: whether this gap is constant is an open question up to this day. Furthermore, \citet[Appendix~A.1]{hartline} study a similar notion to the $\EoR$ and compare it to the $\RoE$ in the context of prior independent mechanism design. Their goal is to measure the performance of an algorithm without knowledge of the input distribution against the best algorithm with full distributional knowledge. We defer the reader to \Cref{app:related_work} for more details on past literature.
\section{Preliminaries}\label{sec:model}

\subsection{Model and Notation}
We consider a setting where there is a ground set of elements $E$, and each element $e\in E$ is associated with a non-negative weight $w_e \sim D_e$. We assume that the distributions have no point masses\footnote{We assume that there are no point masses for simplicity of presentation. All of our theorems can be adjusted to  the case where there are point masses.}, and we denote by $D=\times_{e\in E} D_e$ the product distribution. The elements are presented  with their weights in an online fashion to  a decision-maker who needs to decide immediately and irrevocably whether to accept the current element or not. The decision-maker must ensure that the set of selected elements belongs to a predefined family of downward-closed feasibility constraints $\cF$ at all times. The goal of the decision-maker is to maximize the weight of the selected set. 
We make use of the following definitions and notations.

\begin{definition}\label{def:f}
Let $\bw \in \mathbb{R}_{\geq 0}^{|E|}$ be a non-negative weight vector. We define $\OPT: \mathbb{R}_{\geq 0}^{|E|} \to \cF$ to be the function mapping a vector of weights to a maximum-weight set in family $\cF$. Namely,
\begin{align*}
    \OPT(\bw) = \arg\max_{S \in \cF}{\sum_{e \in S}{w_e}}.
\end{align*}
Moreover, we abuse notation of vector $\bw$ and use $\bw(S)$ to denote the sum of weights in set $S$, i.e.,
\begin{align*}
    \bw(S) := \sum_{e \in S}{w_e}.
\end{align*}
\end{definition}

\begin{definition}\label{fct}
Given a downward-closed family of feasible sets $\cF$, we define $f_{\cF}: \mathbb{R}_{\geq 0}^{|E|} \rightarrow \mathbb{R}$ to be the function that, given a weight vector $\bw$, returns the maximal weight of a feasible set in $\cF$, i.e.,  $f_{\cF}(\bw) := \bw(\OPT(\bw))$.
When clear from context, we omit $\cF$ and use $f$ instead of $f_{\cF}$.
\end{definition}

Given an online algorithm $\ALG$, we denote the (possibly random) set chosen by it given an input $\bw$ by $\ALG(\bw)$. We will denote by $a_{\ALG}(\bw) := \bw(\ALG(\bw))$ the weight of the feasible set chosen by the algorithm for a specific realization $\bw$, and when clear from context, we omit  $\ALG$ from the notation and use $a(\bw)$ instead of $a_{\ALG}(\bw)$. 
Our objective is to design algorithms that maximize the expected ratio between what the online algorithm gets, and the weight of the offline optimum. To measure our performance, given a downward-closed family $\cF$, a product distribution $D$, and an algorithm $\ALG$, we define
\begin{align*}
    \EoR(\cF,D,\ALG) := \E\left[\frac{\bw(\ALG(\bw))}{f(\bw)}\right],
\end{align*}
where the expectation runs over the stochastic generation of the input, as well as the (possible) randomness of the algorithm.
Similarly, we define 
\begin{equation}
    \EoR(\cF,D) := \sup_{\ALG} \EoR(\cF,D,\ALG), \label{eq:EORD}
\end{equation}
and
\begin{equation}
    \EoR(\cF):= \inf_{D}  \EoR(\cF,D). \label{eq:EORF} 
\end{equation}

We will compare our results to the standard objective of maximizing the ratio of expectations between the algorithm and the offline optimum. Accordingly, we denote
\begin{align*}
    \RoE(\cF,D,\ALG) := \frac{\E\left[\bw(\ALG(\bw))\right]}
    {\E\left[{f(\bw)}\right]}.
\end{align*}

The final benchmark we will compare our results to is the probability of selecting an optimal (offline) set, defined by 
\begin{align*}
    \PbM(\cF,D,\ALG) := \Pr\left[\bw(\ALG(\bw)) = f(\bw)\right].
\end{align*}
Analogously to Equations~\eqref{eq:EORD},~and \eqref{eq:EORF}, we define $\RoE(\cF,D),\, \RoE(\cF),\,  \PbM(\cF,D),$ and $ \PbM(\cF)$.

\subsection{Structural Properties}

As a first step before stating and proving the main results, we derive several properties of $f$ that may be of interest beyond this paper. The main technical tool that we use throughout to guarantee only a constant-factor loss in the reduction is the self-bounding property of the (normalized) offline optimum. Since this property resembles a ``smoothness'' condition when removing one of the coordinates of the input vector, we can only prove it if we restrict the support of the weights.



\begin{restatable}[Self-bounding functions]{definition}{bounding}\label{def:bounding}
    Let $\bx := (x_1, \dots, x_n)$ be a vector of independent random variables, and $\cX$ the corresponding product space. Similarly, let $\bx^{(i)} := (x_1, \dots, x_{i-1}, x_{i+1}, \dots, x_{n})$ be the same vector deprived of the $i^{\text{th}}$ coordinate, and $\cX^{(i)}$ the corresponding product space. A function $g: \cX \rightarrow \mathbb{R}$ is said to be self-bounding if there exists a series of functions $\{g_i\}_{1\leq i\leq n}$, such that each $g_i: \cX^{(i)} \rightarrow \mathbb{R}$ satisfies 
    \begin{align*}
        0 \leq g(\bx) - g_i(\bx^{(i)}) &\leq 1,\\
        \sum_{i \in [n]}{\left(g(\bx) - g_i(\bx^{(i)})\right)} &\leq g(\bx).
    \end{align*}
\end{restatable}

\begin{proposition}[Properties of $f_{\cF}$]\label{prop:bounding}
    For every downward-closed family of sets $\cF$, the function $f$ $(=f_\cF )$ satisfies the following properties:
    \begin{enumerate}
        \item $f$ is $1$-Lipschitz.
        \item $f$ is monotone, i.e., if $\bu \geq \bv$ point-wise, then $f(\bu)\geq f(\bv)$.
        \item For every $\tau > 0$, the function  $f/\tau$ restricted to the domain $[0,\tau]^{|E|}$ is self-bounding.
    \end{enumerate}
\end{proposition}

In \Cref{app:xos}, we generalize the above proposition to arbitrary (extended) XOS functions. 

The main attribute of self-bounding functions that we will use is the following inequality.

\begin{theorem}[BLM Inequality \citep{blm2}]\label{blm}
For a self-bounding function $g: \cX \to \mathbb{R}$, it holds that:
\begin{eqnarray*}
    \mbox{\emph{For every }} z>0, \qquad & \Pr\left[g(\bx) \geq \emph{\E}[g(\bx)] + z \right] &\leq e^{-\frac{3z^2}{6\emph{\E}[g(\bx)] + 2z}}.\\
    \mbox{\emph{For every }} z < \emph{\E}[g(\bx)], &
    \Pr\left[g(\bx) \leq \emph{\E}[g(\bx)] - z \right] &\leq e^{-\frac{z^2}{2\emph{\E}[g(\bx)]}}.
\end{eqnarray*}
\end{theorem}


\section{From \RoE { }to \EoR: Feasibility-Based Reduction  
} \label{sec:relation}
Before presenting our reduction from \RoE ~to \EoR, we start with a few definitions and observations. We defer their proofs, as well as other auxiliary claims, to \cref{app:omitted}.
Fixing a parameter $\gamma \in (0,1)$, we define the threshold  $\tau$ given a set of elements $E$ with corresponding distributions $\{D_e\}_{e\in E}$, to be such that 
\begin{equation}
    \Pr\left[\tau \geq \max_{e \in E}{w_e}\right] = \gamma. \label{eq:t}
\end{equation}
Such a $\tau$ exists and is unique for every $\gamma$ since we assume that there are no point masses.
For every $e\in E$, we denote by $\overline{D}_e$ the distribution $D_{e \mid w_e \leq \tau}$, as per the $\tau$ defined in Equation~\eqref{eq:t}. 
This is well defined since $\gamma>0$, and, therefore, the probability that $w_e \leq \tau$ is at least $\gamma >0$.
Given a realization of $\bw$, let $\overline{\bw} \in \mathbb{R}_{\geq 0}^{|E|}$ be the weight vector determined by the following process: For each $e$, if $w_e \leq \tau$ then $\overline{w}_e =w_e$; otherwise, let $\overline{w}_e$ be a fresh (independent) draw from $\overline{D}_e$.   Note that the distribution of $\overline{w}$ is a product distribution, where for each $e \in E$, $\overline{w}_e$ is distributed according to $\overline{D}_e$.


We next define the two events that we will use in our analysis.
\begin{definition}\label{events}
Let us define the following events.
\begin{enumerate}
    \item \textbf{Core.} $\cE_0 := \left\{\forall e \in E: w_e \leq \tau \right\}$.
    \item \textbf{Tail.} $\cE_1 := \left\{\exists ! e \in E: w_e > \tau\right\}$,
\end{enumerate}
where the symbol ``$\exists !$'' signifies the existence of a unique such element.
\end{definition}

The next observation enables us to flexibly change whenever needed from conditioning on $\cE_0$ to working directly with the truncated distribution, and vice versa.

\begin{observation} \label{obs:dists}
The distribution of $\overline{D}$ is identical to the distribution of $D$ conditioned on event $\cE_0$.
\end{observation}


We are now ready to present our reduction from  $\RoE$ to $\EoR$. 
\begin{theorem}\label{alg-analysis}
For every downward-closed family of feasibility constraints $\cF$, it holds that
\begin{align}
\EoR(\cF) \geq \frac{\RoE(\cF)}{12}.
\end{align}
\end{theorem}

Note that in the reduction of \Cref{general}, part of the input is a subroutine $\AALG$ that has an $\RoE$ at least as large as $\alpha$. The condition on the event $\cE_0$ just means that all we need to know is that this $\alpha$-guarantee holds when all weights are below the chosen threshold $\tau$. Starting from $\AALG$, we design an algorithm that uses $\AALG$ as a black box in one of the two cases and achieves an $\EoR$ which is at most a multiplicative constant factor away from the $\RoE$.


\RestyleAlgo{ruled}

\begin{algorithm}
\caption{$\RoE\text{-to-}\EoR$}\label{general}
\KwData{Ground set $E$, distributions $D_e$, feasibility family $\cF$, a subroutine $\AALG$ }
\textbf{Parameters:} $\gamma \in (0,1)$, $c > 0$

\textbf{Assumption:} $\AALG$ satisfies $\E\left[a_{\AALG}(\bw)  \mid \cE_0 \right]  \geq \alpha \cdot  \E\left[f(\bw) \mid  \cE_0\right]$

\KwResult{Subset $\ALG(\bw) \subseteq E$ such that
$\ALG(\bw) \in \cF$}
Calculate $\tau$ according to Equation~\eqref{eq:t}\;
Calculate $W:=\E\left[f(\overline{\bw})\right]$\;

\eIf{$W \leq c \cdot \tau$}
{
    Return the first element $e^* \in E$ such that $w_{e^*} \geq \tau$\;
    If no such element exists, return $\emptyset$\;
}{
    As long as $w_e \leq \tau$ run
    $\AALG(\bw)$\;
    If $w_e > \tau$, reject all remaining elements\;
}
\end{algorithm}

In order to prove \Cref{alg-analysis}, by the definition of $\RoE(\cF)$, we will assume the existence of an algorithm $\AALG$ that satisfies:
\begin{align}
    \E_{\bw \sim D}\left[a_{\AALG}(\bw)  \mid \cE_0 \right] &= \E_{\overline{\bw} \sim \overline{D}}\left[a_{\AALG}(\bw)   \right]  \nonumber \\
    &\geq \alpha \cdot  \E_{\overline{\bw} \sim \overline{D}}\left[f(\overline{\bw}) \right] = \alpha \cdot  \E_{\bw \sim D}\left[f(\bw) \mid  \cE_0\right], \label{eq:assumption}
\end{align} 
where the equalities hold by Observation \ref{obs:dists}.
Our analysis distinguishes between two cases according to whether  $W:= \E\left[f(\overline{\bw})\right] \leq c \cdot \tau$.
\Cref{newcase1} analyzes the case where $W\leq c \cdot \tau$, and \Cref{newcase2} the case where $W > c \cdot \tau$. In the remainder, for ease of notation, we use $a(\bw)$ instead of $a_{\AALG}(\bw)$.

\begin{lemma}[``Catch the superstar'']\label{newcase1}
For all constants $\delta > 1, k \geq 1$ and $c \geq \frac{4 + 2\delta}{3(\delta - 1)^2} \log \frac{k}{\alpha}$, if  $W\leq c \cdot \tau$, then Algorithm~\ref{general} 
satisfies
\begin{align}
    \emph{\E}\left[\frac{a(\bw)}{f(\bw)}\right] \geq \frac{\gamma \log\left(1/\gamma\right)}{c+1}.
\end{align}
\end{lemma}

\begin{proof}
In this scenario, we know that \Cref{general} will select the first element $e \in E$ such that $w_e \geq \tau$, which we denote by $e^*$. It may happen that no such element exists, in which case the algorithm gets a contribution 
of $0$. On the other hand, by Claim \ref{gamma-lemma} (deferred to \Cref{app:omitted}), $\Pr\left[\cE_1\right] \geq \gamma \log\left(1/\gamma\right)$. Then, conditioned on this event, \cref{general} surely (with probability $1$) selects $e^*$ and we have
\begin{eqnarray}
    \frac{a(\bw)}{f(\bw)} &\geq& \frac{w_{e^*}}{w_{e^*} + f(\overline{\bw})} \nonumber\\
    &\geq &\frac{\tau}{\tau + f(\overline{\bw})}, \label{eq:tfrac}
\end{eqnarray}
where the first inequality follows from Claim \ref{opts} (deferred to \Cref{app:omitted}), and the second by observing that since $w_e > \tau$, the ratio is minimized when  $w_e=\tau$. Hence, we get that 
\begin{align*}
    \E\left[\frac{a(\bw)}{f(\bw)}\right] &\geq \E\left[\frac{a(\bw)}{f(\bw)} \given \cE_1\right] \cdot \Pr\left[\cE_1\right]\\
    &\geq \E\left[\frac{\tau}{\tau + f(\overline{\bw})}\right] \cdot \gamma \log\left(1/\gamma\right)\\
    &\geq \frac{\tau}{\tau + W} \cdot \gamma \log\left(1/\gamma\right)\geq \frac{\gamma \log\left(1/\gamma\right)}{c+1},
\end{align*}
where the first inequality is by the law of total expectation,
the second follows from Equation~\eqref{eq:tfrac} and Claim \ref{gamma-lemma}, the third is by Jensen's inequality (see Claim \ref{jensen}), and the last inequality is due to our assumption that $W \leq c \cdot \tau$.
\end{proof}

In the following, \Cref{newcase2}, Claim \ref{blm-opt}, and Claim \ref{cl:eaw} are mainly dedicated to expressing the expected ratio though various manipulations in a convenient form, such that the concentration property of the offline optimum (see \Cref{blm}) can be repeatedly applied. 

\begin{lemma}[``Run the Combinatorial Algorithm'']\label{newcase2}
For all constants $\delta > 1, k \geq 1$ and $c \geq \frac{4 + 2\delta}{3(\delta - 1)^2} \log \frac{k}{\alpha}$, if $W > c \cdot \tau$, then \Cref{general} satisfies 
\begin{align}
    \emph{\E}\left[\frac{a(\bw)}{f(\bw)}\right] \geq \frac{\gamma}{\delta}\frac{k-\delta}{k}\alpha.
\end{align}
\end{lemma}

To prove \Cref{newcase2}, we will make use of the two following claims.

\begin{claim}\label{blm-opt}
For all constants $\delta > 1, k \geq 1$ and $c \geq \frac{4 + 2\delta}{3(\delta - 1)^2} \log \frac{k}{\alpha}$, if $W > c \cdot \tau$, then we have
\begin{align*}
    \Pr\left[f(\bw) > \delta W \mid \cE_0\right] \leq \frac{\alpha}{k}.
\end{align*}
\end{claim}

\begin{proof}
By \Cref{blm}, choosing $z = \frac{1}{\tau}\left(\delta - 1\right)W$, we have
\begin{align*}
    \Pr\left[f(\bw) > \delta W \mid \cE_0\right]  = \Pr\left[\frac{f(\bw)}{\tau} > \frac{\delta W}{\tau} \given \cE_0\right]\leq e^{-\frac{3(\delta - 1)^2}{4 + 2\delta} \frac{W}{\tau}} \leq e^{-\frac{3(\delta - 1)^2}{4 + 2\delta}c} \leq \frac{\alpha}{k},
\end{align*}
since $W > c \cdot \tau$ by assumption, and since $c \geq \frac{4 + 2\delta}{3(\delta - 1)^2} \log \frac{k}{\alpha}$.
\end{proof}

\begin{claim} \label{cl:eaw}
For all constants $\delta > 1, k \geq 1$ and $c \geq \frac{4 + 2\delta}{3(\delta - 1)^2} \log \frac{k}{\alpha}$, if $W > c \cdot \tau$, then we have
\begin{align*}
    \E\left[a(\bw) \given f(\bw) \leq \delta W, \cE_0\right] \geq \left(1 - \frac{\delta - 1}{k}\right)\alpha W.
\end{align*}
\end{claim}
\begin{proof}
We know that 
\begin{align}
    \alpha \cdot W \leq \E\left[a(\bw) \given \cE_0\right]&= \E\left[a(\bw) \given f(\bw) \leq \delta W, \cE_0\right] \cdot \Pr\left[f(\bw) \leq \delta W \mid \cE_0\right] \nonumber\\
    &+ \E\left[a(\bw) \given f(\bw) > \delta W, \cE_0\right] \cdot \Pr\left[f(\bw) > \delta W \mid \cE_0\right] \nonumber\\
    &\leq \E\left[a(\bw) \given f(\bw) \leq \delta W, \cE_0\right]  \nonumber \\
    &+ \E\left[f(\bw) \given f(\bw) > \delta W, \cE_0\right] \cdot \Pr\left[f(\bw) > \delta W \mid \cE_0\right], 
    \label{eq:eaw1}
\end{align}
where the first inequality is by Equation~\eqref{eq:assumption},  the equality is by the Law of Total Expectation, and the second inequality is since a probability is bounded by $1$, and $a(\bw) \leq f(\bw)$.

Next, we show that
\begin{align}
    \Pr&[f(\bw) > \delta W \mid \cE_0] \cdot \E[f(\bw) \mid f(\bw) > \delta W, \cE_0] \nonumber\\
    &= \Pr[f(\bw) > \delta W \mid \cE_0] \cdot \frac{1}{\Pr[f(\bw) > \delta W \mid \cE_0]} \cdot \int_{\delta W}^{\infty}{\Pr[f(\bw) > z \mid \cE_0]dz} \nonumber\\
    &= \int_{\delta W}^{\infty}{\Pr[f(\bw) > z \mid \cE_0 ]dz} = \int_{(\delta-1) W}^{\infty}{\Pr[f(\bw) > W + z \mid \cE_0]dz}\nonumber\\
    &= \int_{(\delta-1) W}^{\infty}{\Pr[\frac{f(\bw)}{\tau} > \frac{1}{\tau}(W + z) \mid \cE_0]dz} \stackrel{(1)}{\leq} \int_{(\delta-1) W}^{\infty}{e^{-\frac{3(z/\tau)^2}{6W/\tau + 2z/\tau}} dz} \nonumber\\
    &\stackrel{(2)}{\leq} 
    \int_{(\delta-1) W}^{\infty}{e^{-\frac{3(\delta-1)}{4 + 2\delta}\frac{z}{\tau}} dz} = \frac{4 + 2\delta}{3(\delta-1)}\tau e^{-\frac{3(\delta-1)^2}{4 + 2\delta}\frac{W}{\tau}} \stackrel{(3)}{<}\frac{4 + 2\delta}{3c(\delta-1)} W  e^{-\frac{3(\delta-1)^2}{4 + 2\delta}c} \nonumber\\
    & \stackrel{(4)}{\leq} (\delta-1)\frac{\alpha}{k}W.\label{eq:blm-exp-opt}
\end{align}
Here, (1) follows from applying \Cref{blm}, (2) from noticing that $W \leq z/(\delta-1)$ within the integral limits, (3) from  the fact that $W > c \cdot \tau$, (4) since $c \geq \frac{4 + 2\delta}{3(\delta - 1)^2} \log \frac{k}{\alpha}$.

Combining Equations~\eqref{eq:eaw1} and \eqref{eq:blm-exp-opt}, we get
\begin{align*}
    \E\left[a(\bw) \given f(\bw) \leq \delta W, \cE_0\right] \geq \left(1 - \frac{\delta - 1}{k}\right)\alpha W.
\end{align*}
\end{proof}

\begin{proof}[Proof of \Cref{newcase2}]
We conservatively assume that if there is at least one element with a weight exceeding $\tau$, then the contribution of the algorithm is $0$. By Claim \ref{gamma-lemma}, $\Pr\left[\cE_0\right] = \gamma$.
We then have that:
\begin{align}
        \E\left[\frac{a(\bw)}{f(\bw)}\right] &\geq \E\left[\frac{a(\bw)}{f(\bw)} \given \cE_0\right] \cdot \Pr\left[\cE_0\right] \nonumber \\
        &\geq \E\left[\frac{a(\bw)}{f(\bw)} \given f(\bw) \leq \delta W, \cE_0\right] \cdot \Pr\left[f(\bw) \leq \delta W  \mid \cE_0\right] \cdot \Pr\left[\cE_0\right]\nonumber\\
        &\geq \frac{\E\left[a(\bw) \mid f(\bw) \leq \delta W, \cE_0\right]}{\delta W} \cdot \left(1 - \Pr\left[f(\bw) > \delta W \mid \cE_0\right]\right) \cdot   \gamma\nonumber\\
        &=  \frac{\gamma}{\delta W} \cdot \left(1 - \Pr\left[f(\bw) > \delta W \mid \cE_0\right]\right) \cdot \E\left[a(\bw) \mid f(\bw) \leq \delta W, \cE_0\right], \label{eq:awfw}
\end{align}
We bound $\Pr\left[f(\bw) > \delta W \mid \cE_0\right]$ from above in Claim \ref{blm-opt}, and $\E\left[a(\bw) \mid f(\bw) \leq \delta W, \cE_0\right]$ from below in Claim \ref{cl:eaw}.

Combining Equation~\eqref{eq:awfw} with Claim \ref{blm-opt} and Claim \ref{cl:eaw}, we get that
\begin{align*}
        \E\left[\frac{a(\bw)}{f(\bw)}\right] &\geq \frac{\gamma}{\delta} \left(1 - \frac{\alpha}{k}\right)\left(1 - \frac{\delta - 1}{k}\right)\alpha \geq \frac{\gamma}{\delta}\frac{k-\delta}{k}\alpha,
\end{align*}
which concludes the proof.
\end{proof}

To complete the proof of \Cref{alg-analysis}, we need to carefully choose the parameters of the previous claims and lemmas (including the threshold of \Cref{eq:t}, used for the case distinction and beyond) so that the multiplicative loss from the two cases balances to a (relatively good) constant.

\begin{proof}[Proof of \Cref{alg-analysis}]
If $\alpha=\RoE(\cF)$, then there exists an algorithm $\AALG$ that satisfies Equation~\eqref{eq:assumption}.
\Cref{general} guarantees the minimum between the expected competitive ratios of \Cref{newcase1} and \Cref{newcase2}. Thus, we need to find parameters $\gamma,\delta,k,c $ that maximize the following  constrained optimization problem:
\begin{align*}
    \mbox{ \text{maximize} } \quad \quad &\min\left\{\frac{\gamma \log\left(1/\gamma\right)}{c+1}, \frac{\gamma}{\delta}\frac{k-\delta}{k}\alpha\right\}\\
    \text{subject to} \quad \quad & \enspace \gamma \in (0,1), \delta > 1, k > 2,  c \geq \frac{4 + 2\delta}{3(\delta - 1)^2} \log \frac{k}{\alpha}.
\end{align*}
The only parameter we do not control is $
\alpha$, as it depends on the feasibility structure $\cF$. If we choose $\gamma = 1/2, \delta = 2, k = 3, c = \frac{8}{3}\log\frac{3}{\alpha}$ all the constraints are satisfied, and we  get
\begin{align*}
    \EoR \geq \min\left\{\frac{\frac{1}{2} \log 2}{\frac{8}{3} \log \frac{3}{\alpha} +1}, \frac{\alpha}{12}\right\} = \frac{\alpha}{12}.
\end{align*}
This concludes the proof.
\end{proof}
\section{From \EoR { }to \RoE: Instance-Based Reduction}\label{sec:reduction2}
In this section, we show an instance-based reduction from  $\EoR$ to $\RoE$. Unlike \Cref{alg-analysis}, our next result shows that the $\RoE$ is always at least a constant fraction of the $\EoR$, for every pair of (downward-closed) feasibility constraint $\cF$ and product distribution $D$. We will assume the existence of an algorithm $\ALGEOR(\bw)$ that satisfies $\E_{\bw \sim D}\left[a_{\ALGEOR}(\bw)/f(\bw)\right] \geq \alpha$. For ease of notation, we denote the  value of \Cref{alg:eor-roe} on $\bw$ by $ a(\bw)$.

\RestyleAlgo{ruled}
\begin{algorithm}
\caption{$\EoR\text{-to-}\RoE$}\label{alg:eor-roe}
\KwData{Ground set $E$, distributions $D_e$, feasibility family $\cF$, a subroutine $\ALGEOR$ }
\textbf{Assumption:} $\ALGEOR(\bw)$ satisfies that $\E_{\bw \sim D}\left[a_{\ALGEOR}(\bw)/f(\bw)\right] = \alpha$


\KwResult{Subset $\ALG(\bw) \subseteq E$ such that
$\ALG(\bw) \in \cF$}
Let $A := \E_{w \sim D}[\max_{e \in E} w_e]$\;

\eIf{$A \geq \alpha \cdot 
\E_{\bw\sim D}[f(\bw)]/34$}
{
    Return the first element $e^* \in E$ such that $w_{e^*} \geq \frac{A}{2}$\;
    If no such element exists, return $\emptyset$\;
}{
    Run
    $\ALGEOR(\bw)$\;
}
\end{algorithm}

\begin{theorem}\label{alg-analysis-rev-new}
For every downward-closed family of feasibility constraints $\cF$, and every product distribution $D$ it holds that
\begin{align}
\RoE(\cF,D) \geq \frac{\EoR(\cF,D)}{68}.
\end{align}
\end{theorem}

\begin{proof}
Let $\alpha = \EoR(\cF,D)$, and let $A = \E_{w \sim D}[\max_{e \in E} w_e]$ (as per \Cref{alg:eor-roe}'s pseudocode). Our algorithm has two cases, depending on the value of $A$.

In the first case,  $A \geq \alpha \cdot \frac{\E[f(\bw)]}{34}$: \Cref{alg:eor-roe} sets a threshold of $\frac{A}{2}$ and selects the first element that exceeds it. The algorithm achieves
\begin{align*}
    \E\left[a(\bw)\right] \geq \frac{A}{2} \geq \alpha \cdot \frac{\E[f(\bw)]}{68},
\end{align*}
where the first inequality follows by the Prophet Inequality \citep{samuel-cahn}, and the last from the assumption on $A$ in this case.

Otherwise, we have that $A < \alpha \cdot \frac{\E[f(\bw)]}{34}$, in which case \Cref{alg:eor-roe} simply runs the $\ALGEOR$ subroutine as a blackbox. Before proceeding, we show the following claim, which will prove useful in the remainder of this proof. To this end, let $w_e' = w_e \cdot \bone_{w_e\leq 2A}$, for every $e \in E$, let $D_e'$ be the distribution of $w_e'$, and let $D'$ be their product distribution. 
\begin{claim}\label{cl:3}
It holds that $\E_{\bw\sim D}[\sum_{e \in E} w_e-w_e'] \leq 2\alpha \cdot \frac{\E_{w\sim D} [f(w)]}{34}$.  
\end{claim}
\begin{proof}
Let event $B_e = \{w_{e'} \leq 2A, \forall e'\neq e\}$. We have that
\begin{align}
    \E\left[\sum_{e \in E} w_e-w_e'\right] &= \E\left[\sum_{e \in E} w_e \cdot \bone_{w_e > 2A}\right] \leq 2 \cdot \E\left[\sum_{e \in E} w_e \cdot \bone_{w_e > 2A} \cdot \bone_{B_e} \right] \leq 2A, \label{eq:2a}
\end{align}
where the first inequality follows from noting that (1) $B_e$ is independent of the realization of $w_e$, and (2) $\Pr[B_e] \geq \frac{1}{2}$ since $\Pr[B_e] \geq \Pr[\max_{e \in E} w_e \leq 2A]$ and by applying Markov's Inequality on random variable $\max_{e \in E} w_e$. The second inequality follows from the fact that we are accounting for selecting the maximum only when there  is a unique  element whose weight exceeds $2A$ (which is at most the expected maximum). The claim then holds by combining the last inequality and the assumption on $A$ from being in this case.
\end{proof}

With this claim at hand, we can continue with the proof of the theorem. By \Cref{prop:bounding}, the function $\frac{f(\bw)}{2A}$ is self-bounding in the domain $[0,2A]^{|E|}$.
 Let  $Q = \frac{\E_{\bw' \sim D'} [f(\bw')]}{2A} $.  It holds that
 \begin{align}
     Q &= \frac{\E_{\bw' \sim D'} [f(\bw')]}{2A} > \frac{34 \cdot \E_{\bw' \sim D'} [f(\bw')]}{2\alpha \cdot \E_{\bw \sim D} [f(\bw)]} \nonumber\\
     &\geq \frac{17 \cdot \left(\E_{\bw \sim D} [f(\bw)] - 2A\right)}{\alpha \cdot \E_{\bw \sim D} [f(\bw)]} > \frac{17}{\alpha}\left(1 - \frac{2\alpha}{34}\right) \geq \frac{16}{\alpha} \label{eq:Q},   
 \end{align}
 where the first and third inequalities follow simply by using $A < \alpha \cdot \frac{\E[f(\bw)]}{34}$. The second inequality follows by Claim \ref{cl:3}, by Lipschitzness of $f$, Inequality~\eqref{eq:2a} and the value of $A$. The last inequality is since  $\alpha 
 \leq 1$. We, therefore, have that
 \begin{equation}
     \Pr_{\bw' \sim D'}\left[\frac{f(\bw')}{2A} \leq \frac{Q}{2}  \right] \leq e^{-Q/8} \stackrel{\eqref{eq:Q}}{\leq} e^{-2/\alpha} \leq \frac{\alpha}{2}, \label{eq:alpha2}
 \end{equation} 
 where the first inequality follows from \Cref{blm}, and the last inequality is by the identity $e^{-x} \leq 1/x$ for all $x > 0$. Let $V = \E_{\bw\sim D}[f(\bw)]$, 
 then it holds that 
 \begin{equation}
     \frac{V}{4} \leq \frac{\E_{\bw \sim D} [f(\bw) ]- 2A}{2} \leq \frac{\E_{\bw' \sim D'} [f(\bw')]}{2} = Q\cdot A, \label{eq:qw}
 \end{equation}
 where the first inequality is by using $A < \alpha \cdot \frac{\E[f(\bw)]}{34}$, and the second inequality is again by Claim \ref{cl:3}, by Lipschitzness of $f$, Inequality~\eqref{eq:2a} and the value of $A$.
Thus,
\begin{equation}
        \Pr_{\bw\sim D}\left[f(\bw) \leq \frac{V}{4}  \right]  \leq \Pr_{\bw'\sim D'}[f(\bw') \leq  Q \cdot A ] \stackrel{\eqref{eq:alpha2}}{\leq}  \frac{\alpha}{2}, \label{eq:alpha22}
\end{equation}
where the first inequality holds since $f$ is monotone, since $\bw'\leq \bw$ component-wise, and since by Equation~\eqref{eq:qw}, $\frac{V}{4}\leq  Q\cdot A$. 
We next bound $\E_{\bw\sim D} \left[\frac{a(\bw)}{f(\bw)} \mid f(\bw) \geq \frac{V}{4}\right]$. By definition of $\alpha$, it holds that 
\begin{align*}
 \alpha & \leq  \E_{\bw\sim D} \left[\frac{a(\bw)}{f(\bw)}\right] \\
 &=\E_{\bw\sim D} \left[\frac{a(\bw)}{f(\bw)} \mid f(\bw)<\frac{V}{4}\right] \cdot \Pr\left[f(\bw)<\frac{V}{4}\right] \\
 &+\E_{\bw\sim D} \left[\frac{a(\bw)}{f(\bw)} \mid f(\bw) \geq \frac{V}{4}\right] \cdot \Pr\left[f(\bw)\geq \frac{V}{4}\right] \\
 &\leq \frac{\alpha}{2}  +\E_{\bw\sim D} \left[\frac{a(\bw)}{f(\bw)} \mid f(\bw) \geq \frac{V}{4}\right],
\end{align*}
which by rearranging, we get that 
\begin{equation}
\E_{\bw\sim D} \left[\frac{a(\bw)}{f(\bw)} \mid f(\bw) \geq \frac{V}{4}\right] \geq \frac{\alpha}{2}. \label{eq:alpha222}
\end{equation}
This implies that 
\begin{align*}
    \RoE(\cF,D) & \geq \E_{w\sim D} [a(\bw)] / V \geq \E_{\bw\sim D} \left[a(\bw) \mid f(\bw) \geq \frac{V}{4}\right] \cdot \Pr\left[f(\bw) \geq \frac{V}{4} \right]/V \\
    &\geq \frac{\alpha\cdot V}{8}\cdot \left(1-\frac{\alpha}{2}\right)/V \geq \frac{\alpha}{16},
\end{align*}
where the third inequality holds by \Cref{eq:alpha22,eq:alpha222}, and the last inequality is since $\alpha \leq 1$.
This concludes the proof.
\end{proof}

An immediate corollary of \Cref{alg-analysis-rev-new} is:
\begin{corollary}
For every downward-closed family of feasibility constraints $\cF$ it holds that
\begin{align*}
\RoE(\cF) \geq \frac{\EoR(\cF)}{68}.
\end{align*}
\end{corollary}
The above corollary and  \Cref{alg-analysis} imply
\begin{corollary}\label{col:eorroe}
For every downward-closed family of feasibility constraints $\cF$ it holds that
\begin{align*}
\frac{\RoE(\cF)}{\EoR(\cF)} \in \Theta(1).
\end{align*}
\end{corollary}
\section{From \RoE~to \EoR: an Impossibility of Instance-Based Reduction}\label{sec:impossibility}
In this section, we show that an instance-based reduction from $\RoE$ to $\EoR$ is unachievable. This is in contrast to the reduction of \Cref{alg-analysis-rev-new} from $\EoR$ to $\RoE$. In particular, we show that there are  a feasibility constraint  $\cF$ and a product distribution $D$ for which $\RoE(\cF,D)$ is constant but $\EoR(\cF,D)$ is sub-constant. We show this using the following stronger claim: 
\begin{proposition}\label{prop:reduction}
For every feasibility constraint $\cF$  there exists a product distribution $D$ such that  $\EoR(\cF,D)  \in O\left(\RoE(\cF)\right)$ while  $\RoE(\cF,D)  \in \Omega\left(1\right)$. 
\end{proposition}

\begin{proof}
    Let $\RoE(\cF) = x$. By definition of $\RoE(\cF)$ as an infimum over all product distributions $D$, there exists a product distribution $D$ such that $\RoE(\cF, D) \leq 
    \frac{69x}{68}$. We now consider the product distribution $D'$ constructed from leaving all elements' distributions unaltered but modifying the distribution $D_{e^*}$ of an arbitrary element $e^*$ into $D'_{e^*}$ as follows. Consider the random variable $w_{e^*}$ representing the weight of element $e^*$ and let 
    \begin{align*}
        w'_{e^*} := w_{e^*} + \frac{\E\left[\sum_{e \in E} w_e\right]}{x} \cdot X,
    \end{align*}
    where $X \sim \text{Bernoulli}(x)$. It is easy to see that $\RoE(\cF, D') \geq \frac{1}{2}$, since  the algorithm that always selects element $e^*$ obtains $\E[w'_{e^*}]\geq \E\left[\sum_{e \in E} w_e\right]$, while the prophet can achieve at most $2\cdot \E\left[\sum_{e \in E} w_e\right]$. To conclude, we have that
    \begin{align*}
       \EoR(\cF, D') & \leq  \Pr[X=1] \cdot 1 + \Pr[X=0] \cdot\EoR(\cF, D)  \\ & \leq x +\EoR(\cF, D) \leq x + 68 \cdot \RoE(\cF, D) \leq 70 x,
    \end{align*}
    where the first inequality above derives from the fact that, if the Bernoulli random variable $X$ is $1$, we upper bound the algorithm's performance by that of the optimum; otherwise, we upper bound the performance by $\EoR(\cF, D)$, the third inequality follows from \Cref{alg-analysis-rev-new}, and the last by recalling that $\RoE(\cF, D) \leq \frac{69x}{68}$. 
\end{proof}

We know by the example presented in Appendix B of \citep{rubinstein-matr} that is based on an example from \citep{BabaioffIK07} for a different setting, that there exists a feasibility constraint with $n$ elements such that $\RoE\left(\cF\right) \in O\left(\frac{\log\log(n)}{\log(n)}\right)$.
Combining \cref{prop:reduction} with this example for large enough $n$ implies that:
\begin{corollary}\label{cor:noreduction}
For every $\epsilon>0$, there exist a feasibility constraint $\cF$  and a product distribution $D$ in which $\EoR(\cF,D)  \leq \epsilon $ and  $\RoE(\cF,D)  \in \Omega(1)$.
\end{corollary}

\section{Discussion}\label{sec:discussion}

In this paper, we studied the performance of combinatorial prophet inequalities via the expected ratio ($\EoR$). We focus on its connections to the standard measure of performance in the literature, i.e., the ratio of expectations ($\RoE$). We establish that, for every downward-closed feasibility constraint, the gap between $\RoE(\cF)$ and $\EoR(\cF)$ is at most a constant. Moreover, we show that the $\EoR$ is an even stronger benchmark in the sense that $\RoE(\cF, D)$ is at least a constant of $\EoR(\cF, D)$, but not vice versa.

We want to remark that \Cref{general} and \Cref{alg:eor-roe} are constructive ways to prove \Cref{alg-analysis} \Cref{alg-analysis-rev-new}. For example, the purpose of \Cref{general} is to show that for every family of feasibility constraints $\cF$, $\EoR(\cF) \geq \RoE(\cF)/12$. For $\alpha=\RoE(\cF)$, by definition of $\RoE(\cF)$, there exists an algorithm that satisfies the assumption made in the algorithm, and for the proof of the theorem, this can be used to show an existence of an algorithm with $\EoR(\cF) \geq \alpha/12$. Therefore, the assumption simply restates the starting point of the reduction of \Cref{alg-analysis}. 

In the remainder of this section, we state some remarks and discuss extensions that follow from our proofs and techniques.

\paragraph{Arrival order.} We first note that our results hold for any arrival order of the elements, such as random \citep{esfandiari-ps}, free \citep{Yan11}, or batch arrival order \citep{ezra}. 

\paragraph{Single-sample.} We can also consider scenarios in which the decision-maker does not have full knowledge of the distributions of the elements' weights. In fact, our reductions can be adjusted (with slightly worse constants) to scenarios in which the decision-maker has only a single sample from each distribution (see \Cref{app:ss} for a formal discussion). 

\paragraph{Extension to XOS functions.}
Our results extend beyond additive functions over downward-closed feasibility constraints, namely to extended-XOS functions (see \Cref{app:xos} for a formal discussion). 
We generalize  the optimum function (\Cref{def:f}) to be $\OPT(\bw) := \underset{S \in \cF}{\arg\max} ~\underset{i \in [\ell]}{\max} \ev{b_i, \bw_S}$, where $\bw_S \in \mathbb{R}^{|E|}$ is the vector of elements weights in $S$ (and $0$ for elements not in $S$), while each $b_i \in \mathbb{R}^{|E|}$ is a vector of nonnegative coefficients. Similarly to the additive case, we have that $f(\bw) = \max_{\substack{i \in [\ell]\\S \in \cF}} \ev{b_i, \bw_S}$. We can show that, by setting $a^S_{ij} := b_{ij} \cdot \bone_{j \in S}$, $f(\bw)$ can be expressed as $f(\bw) = \max_{\substack{i \in [\ell]\\S \in \cF}}\ev{a^S_i, \bw_S}$. Moreover, the functions resulting from projecting all such $f$'s onto $\{0,1\}^{|E|}$ not only are XOS but describe all XOS functions. We can now run ~\Cref{general}, and perform the case distinction with the modified threshold $\tau$ being such that $\Pr\left[\exists j \in [|E|]: w_j > \frac{\tau}{\max_{i \in [\ell]}a_{ij}}\right] = \gamma$. With this at hand, the ``catch the superstar'' subroutine becomes selecting the first element $j$ with $w_j > \frac{\tau}{\max_{i \in [\ell]}a_{ij}}$, while the ``run the combinatorial algorithm'' remains unaltered (up to slight modifications, described in \Cref{app:xos}). For \Cref{alg:eor-roe}, we redefine $A := \E[\max_{i,j} a_{ij} \cdot w_j]$ and repeat a similar analysis to the one above. All in all, we lose an additional $\max_{i,j}{a_{ij}}$ factor in the expected ratio of \Cref{alg-analysis} and \Cref{alg-analysis-rev-new}, and we get that:
\begin{align*}
    \EoR(\cF) \geq \frac{\RoE(\cF)}{12 \cdot \max_{i,j}{a_{ij}}}, \quad \mbox{and} \quad  \RoE(\cF, D) \geq \frac{\EoR(\cF, D)}{68 \cdot \max_{i,j}{a_{ij}}}.
\end{align*}

\paragraph{Other measures of performance.} 

In this work, we study the expected ratio as our measure of performance. Another natural performance measure is the expected inverse ratio, i.e., $\EoIR(\cF) := \E\left[f(\bw)/a(\bw)\right]$. We now show that such a measure may be unbounded even for the single-choice case. Let us consider again \cref{ex:ex3}. Fix a (randomized) algorithm that selects the first element with probability $p$, and let $a_p(\bw)$ be its performance on input $\bw$. Then we have
\begin{align*}
    \EoIR(\cF) & \geq \min_p \E\left[\frac{f(\bw)}{a_p(\bw)}\right] \\
    &\geq \min_p (1-\varepsilon) \cdot \left(p \cdot \frac{1}{1} + (1-p) \cdot \frac{1}{\varepsilon^2}\right) + \varepsilon \cdot \left(p \cdot \frac{1/\varepsilon^2}{1} + (1-p) \cdot \frac{1/\varepsilon^2}{1/\varepsilon^2}\right)\\
    &= \min_p  \frac{1-\varepsilon}{\varepsilon^2} + \varepsilon - \left(\frac{1-\varepsilon}{\varepsilon^2} - \frac{1}{\varepsilon} - 1 + 2\varepsilon\right)p \geq \frac{1}{\varepsilon} + 1 - \varepsilon > \frac{1}{\varepsilon}.
\end{align*}
Hereby, the second inequality holds since the algorithm selects the second box (given it is realized) with probability at most $1-p$, while the third inequality follows from setting $p=1$ to minimize the expression. As $\varepsilon$ can be made arbitrarily small, we have an unbounded $\EoIR$.

In spite of the above simple impossibility result in maximization problems, 
the same measure of performance could be of use when the decision-maker seeks to minimize a function subject to, e.g., covering constraints with stochastic inputs. As mentioned in \Cref{relwork}, \citet{GargGLS08} study the relation between $\RoE$ and $\EoR$ for minimization problems, such as Online Steiner Tree and Traveling Salesman Problem with stochastic inputs. On the other hand, the $\EoIR$ measure in this setting remains unexplored. It would be interesting to understand whether the reductions provided in \Cref{sec:relation,sec:reduction2} are generalizable to the minimization setting (both for the $\EoR$ and the $\EoIR$). 

\paragraph{Gap between \EoR { }and \RoE}  Despite the similarity between the benchmarks, it is not at all obvious whether the maximal gap between $\RoE$ and $\EoR$ for prophet settings and every downward-closed feasibility constraint would be constant (it is an open question whether this gap is constant in other Bayesian settings \cite{GargGLS08}). In our feasibility-based reduction from $\RoE$ to $\EoR$, we lose a constant of $\frac{1}{12}$. Losing a constant is unavoidable already from the single-choice setting, where there is a (tight) gap of $\frac{2}{e}$. It would be interesting to study whether this is the worst gap possible. In the other direction (i.e., from $\EoR$ to $\RoE$), the gap in the reduction is also not tight; it is even possible that the $\RoE(\cF)$ is at least the $\EoR(\cF)$ for every feasibility constraint $\cF$. However, we know that $\RoE(\cF, D)$ can be smaller than $\EoR(\cF, D)$ by a factor of $2$ by \Cref{ex:ex1} and can be unboundedly larger than $\EoR(\cF, D)$. Additionally, finding the exact value of $\EoR(\cF)$ for specific downward-closed constraints (e.g., matching, matroid, knapsack, etc.) is an interesting open question.

\section*{Acknowledgements}
Partially supported by the ERC Advanced Grant 788893 AMDROMA ``Algorithmic and Mechanism Design Research in Online Markets'' and MIUR PRIN project ALGADIMAR ``Algorithms, Games, and Digital Markets''. The last author further acknowledges the support of the Alexander von Humboldt Foundation with funds from the German Federal Ministry of
Education and Research (BMBF), the Deutsche Forschungsgemeinschaft (DFG, German Research Foundation) - Projektnummer 277991500, the COST Action CA16228 “European Network for Game Theory” (GAMENET), and ANID, Chile, grant ACT210005. Most of this work was done while the author was at TU Munich and Universidad de Chile. The views expressed in this paper are the author's and do not necessarily reflect those of the European Central Bank or the Eurosystem.
Tomer was also supported by the National Science Foundation under Grant No. DMS-1928930 and by the Alfred P. Sloan Foundation under grant G-2021-16778, while the author was in residence at the Simons Laufer Mathematical Sciences Institute (formerly MSRI) in Berkeley, California, during the Fall 2023 semester.

\bibliography{references}
\appendix

\section{Further Related Work}\label{app:related_work}

In the prophet inequalities literature, a common underlying assumption is that the decision-maker has full knowledge of the distributions from where the values of the arriving elements are drawn. This is arguably a strong assumption in many practical applications; therefore, a parallel line of work has focused on settings where the distributions are unknown, but a limited number of samples from these distributions is available to the decision-maker.  \citet{azar} pioneered this idea and showed positive results in several combinatorial settings. In fact, in the classic prophet inequality, just a single sample from each distribution suffices to recover the tight result with full distributional knowledge \citep{rubinstein}. Similar insights are obtained when the distributions are identical; this problem was initially studied by \citet{correa} and subsequently improved several times \citep{rubinstein,googol,kaplan,correa2020sampledriven,correa_iid2}. \citet{azar,caramanis,kaplan} extended the single-sample framework (which can be viewed as the minimum amount of available information) to several combinatorial problems.

Directly related to prophet inequalities, a limited number of recent papers consider different performance metrics; we can loosely divide them into two main categories. In the first one \citep{anari,niazadeh,papadimitriou,ezra2023significance,braverman,duetting2023PS}, the goal is to compare against the computationally unbounded optimal online policy. Similar in spirit is the work of \citep{agrawal}, where the decision-maker can choose the ordering of the elements, and the main question is whether finding the optimal ordering can be done in polynomial time. In general, given an online Bayesian selection problem, the natural questions are whether it is hard to compute an optimal solution and, if that is the case, how well we can approximate this benchmark with polynomial-time algorithms. The second group of papers \citep{esfandiari,googol,nuti,ezra2023significance} studies single-choice problems with the goal of maximizing the probability of picking the element with the highest value. Note that this is the objective of the secretary problem, but in problems with stochastic input (as is the case in prophet inequalities). Alternative measures of performance have also been proposed to capture the behavior of biased (as opposed to rational) agents \citep{kleinberg}, or to address fairness considerations \citep{correa-fairness}.

\section{The Single-Choice Case: \PbM { }and \EoR { }are equivalent}\label{app:single}
In this appendix, we will show that for the single-choice prophet inequality, the $\EoR$ is equivalent to the $\PbM$. The (tight) bound on the $\PbM$ of $1/e$ was already shown in \citet{esfandiari}, and \Cref{prop:pbm-eor-1/e} shows that by similar arguments  the $\EoR$ is also $1/e$. On the one hand, it is immediate to show that the $\EoR$ is at least the $\PbM$. On the other hand, we show how to adapt the tight example of the $\PbM$ to work also with respect to the $\EoR$. Before proceeding, let us note that in single-choice $\EoR = \E\big[\frac{w_{\ALG}}{\max_{e \in E} w_e}\big]$, where $w_{\ALG}$ is the value selected by the algorithm.

\begin{proposition}[$\EoR$-$\PbM$ equivalence]\label{prop:pbm-eor-1/e}
In the single-choice prophet inequality  ($\cF=\{S \subseteq E \mid |S|\leq 1\}$), 
it holds that 
\begin{align*}
    \EoR(\cF) = \PbM(\cF).
\end{align*}
\end{proposition}
\begin{proof}
The direction of $\EoR(\cF) \geq \PbM(\cF)$ is trivial since, if we count only the cases when an algorithm selects the maximum, then $w_{\ALG}/\max_{e \in E} w_e=1$, which happens with probability of at least $\PbM(\cF)$. The other direction follows by considering the following instance: for $i \in [|E|]$, let
\begin{align*}
    w_1 \sim D_1 = 1, \ w_i \sim D_i = \begin{cases}
        0, \text{ w.p. } 1 - \frac{1}{|E|}\\
        M^{i-1}, \text{ w.p. } \frac{1}{|E|}
    \end{cases}, \forall i > 1.
\end{align*}
By Yao's Minimax Principle \citep{yao_minimax}, we assume without loss of generality that the best algorithm (with respect to $\EoR$) is a deterministic algorithm. Moreover, observe that any optimal algorithm can be assumed to never select zeroes and be history-independent. The latter is true because either the algorithm sees $0$ (which should never be selected) or a value that is the maximum so far, and thus, all the preceding values are irrelevant for both the algorithm and the optimum. 
Moreover, for this instance, any deterministic, history-independent algorithm that does not select zeroes, can be described as selecting a value if and only if it belongs to a fixed subset of values $S \subseteq \{1,M,\ldots,M^{|E|-1}\}$. 
We now distinguish between two cases: If $1\in S$, then   $\EoR(\ALG,D,\cF) \leq 1 \cdot (1-\frac{1}{|E|})^{|E|-1} + \frac{1}{M} \approx \frac{1}{e}$, where the approximation is since we can take arbitrarily large $M$ and $|E|$. 
If $1\notin S$, then for every $0<i<j$, such that $M^i\in S$, and $M^j \notin S$, the performance of the algorithm improves by replacing $S$ by $S\setminus \{M^i\} \cup \{M^j\}$. Repeating this argument, we get that $S$ can be described as $\{M^i \mid i\geq k\}$, for some constant $k>0$, which means it can be described by a single-threshold deterministic algorithm.

Note that if the algorithm picks a non-maximal element, the $\EoR \leq 1/M$, which tends to $0$ for $M$ tending to infinity. This implies that any constant competitive algorithm (in terms of expected ratio) has to select the maximum exactly. Let us now denote by $\rho_i$ the probability that the maximum is selected, had the algorithm chosen threshold $\sigma = M^{i-1}$. We have for all $i > 1$,
\begin{align*}
    \rho_i &:= \Pr\left[\frac{w_{\ALG}}{\max_{e \in E}{w_e}} = 1 \given \sigma = M^{i-1}\right] = \sum_{j \geq i}{\left(\frac{1}{|E|} \cdot \prod_{\substack{k > i\\k \neq j}}{\left(1 - \frac{1}{|E|}\right)}\right)}\\
    &= \prod_{j \geq i}{\left(1-\frac{1}{|E|}\right)} \cdot \sum_{k \geq i}{\frac{1/|E|}{1 - 1/|E|}},
\end{align*}
where the first equality is since if the threshold of the algorithm is $\sigma = M^{i - 1}$, then  $\ALG$ selects the maximum 
if  exactly one  among $(w_i,\ldots ,w_{|E|})$ is not $0$. 
We now simplify the above expression and obtain
\begin{align*}
   \rho_1 &:= \left(1 - \frac{1}{|E|}\right)^{|E|-1} \ \text{ and } \ \rho_i = \frac{|E|-i+1}{|E|} \cdot \left(1 - \frac{1}{|E|}\right)^{|E|-i}.
\end{align*}
We note that $\rho_i$ is a decreasing sequence, and is  maximized for $i=1$.   It holds that for every $M,|E|$, $\EoR(\cF) \leq  \frac{1}{M} + \rho_1$, thus, 
when $M,|E|$ are approaching infinity, we get that  
\begin{align*}
    \EoR(\cF) \leq \lim_{M,|E| \rightarrow \infty } \left( \frac{1}{M} + \rho_1\right) =  \frac{1}{e},
\end{align*}
as desired.
\end{proof}

Below, we show an even stronger result for single-choice prophet. For a constant $M>0$, and a product distribution $D$ of dimension $n$, we denote by $M^D$ the following product distribution: a vector $\bw$ is drawn from $D$, and the realized vector is then $M^{\bw}:=(M^{w_1},\ldots,M^{w_n})$.  The next result states that, for the single-choice feasibility constraint ($\cF = \{S \subseteq E \mid |S|\leq 1\}$), and every product distribution $D$,  $\PbM(\cF, D)$ equals to the limit of $\EoR(\cF, M^D)$, for $M$ that goes to infinity. 
\begin{proposition}\label{prop:pbm-eor-all}
In the single-choice prophet inequality  ($\cF=\{S \subseteq E \mid |S|\leq 1\}$), 
it holds that 
\begin{align*}
    \lim_{M \rightarrow \infty}\EoR(\cF, M^D) = \emph{\PbM}(\cF, D).
\end{align*}
\end{proposition}
\begin{proof}
Without loss of generality, let us assume that each distribution $D_e$ is defined over the support $\mathbb{R}_{\geq \varepsilon}$. Furthermore, we discretize the support into $\varepsilon$-sized bins and consider the distribution $D_e^\prime$, resulting from the original distribution when each element in a given bin is associated to the bin's left endpoint. For each $D^\prime_e$, let us define distribution $M^D_e$, whose support is $\left\{M^w : w \in \textsf{support}(D^\prime_e)\right\}$, to be the distribution obtained from drawing $w_e \sim D_e$, and raising $M$ to the power of $w_e$. We observe that $M^D = \times_{e\in E} M^D_e$, and has support $\bigcup_{e \in E}{\left\{M^w : w \in \textsf{support}(D^\prime_e)\right\}}$. 

Suppose that $\ALG$ selects $\max_{e \in E}{w_e} = M^t$ with probability $\PbM(\cF, D)$. Note that for any $s < t$, the ratio $M^s/M^t \leq M^{-\varepsilon}$. We take $M \gg 1/\varepsilon$ and obtain
\begin{align*}
    \lim_{M \rightarrow \infty}\EoR(\cF, M^D) \leq \lim_{M \rightarrow \infty}\PbM(\cF, M^D) + (1 - \PbM(\cF, M^D)) \cdot M^{-\varepsilon} = \PbM(\cF, D).
\end{align*}
This concludes the proof.
\end{proof}

An alternative proof of \Cref{prop:pbm-eor-1/e} follows immediately from \Cref{prop:pbm-eor-all}.  
This provides an alternative proof to the fact that $\EoR$ and $\PbM$ are equivalent in the single-choice setting.
Indeed, we have
\begin{align*}
    \EoR(\cF) = \inf_D\lim_{M \rightarrow \infty}\EoR(\cF, M^D) = \inf_D\PbM(\cF, D) = \PbM(\cF).
\end{align*}

\section{Implications of ``Good'' \EoR}\label{app:const-eor}

In this appendix, we formalize the discussion presented in the introduction, regarding the implications of a good $\EoR$.

\begin{claim}\label{cl:imply}
For a (downward-closed) feasibility constraint $\cF$ and a product distribution $D$, if algorithm $\emph{\ALG}$ satisfies $\EoR(\cF,D,\ALG) \geq \alpha$, 
then
\begin{align*}
    \Pr[a(\bw) \geq \frac{\alpha}{2} \cdot f(\bw)] \geq \frac{\alpha}{2}.
\end{align*}
\end{claim}

\begin{proof}
$ \Pr[a(\bw) \geq \frac{\alpha}{2} \cdot f(\bw)] \geq \EoR(\cF) - \Pr[a(\bw) < \frac{\alpha}{2}f(\bw)] \cdot \frac{\alpha}{2}  \geq \frac{\alpha}{2}. $
\end{proof}

\begin{remark}
We remark that our $\RoE$ to $\EoR$ black-box reductions have stronger guarantees than what Claim \ref{cl:imply} states.
Indeed, \Cref{sec:relation} shows that with constant probability (as opposed to a probability depending on $\alpha$, and thus the feasibility constraint), an $O(\alpha)$ approximation of the ex-post optimum can be achieved.
\end{remark}

\begin{claim}\label{clm:one-over-e}
Even in single-choice settings, no algorithm can select the maximum (or even any constant approximation to the maximum) with probability larger than $\frac{1}{e}$ (see \Cref{app:single} for details).
\end{claim}

\begin{claim} \label{clm:two-thirds}
For every constants $c,\varepsilon>0$, there exist a (downward-closed) feasibility constraint $\cF$ and a product distribution $D$ for which $\EoR(\cF)$ is constant, and such that the probability that any algorithm achieves better than $\left(\frac{2}{3} + \varepsilon\right)$-approximation of the offline optimum is arbitrarily small.
That is,
\begin{align*}
    \Pr[a(\bw) \geq \left(\frac{2}{3}+\varepsilon\right)\cdot f(\bw)] \leq 5e^{-\frac{\varepsilon^2 n}{12}}.
\end{align*}
\end{claim}
\begin{proof}
The proof of this claim derives directly from \Cref{ex:ex2}. Let us first show that the expected ratio is constant (no matter what $D$ is). This is indeed the case because an algorithm that selects either all the first elements or all the second elements of the pairs with probability $\frac{1}{2}$ each, will achieve an $\EoR(\cF) \geq \frac{1}{2}$. 

We now show the second part of the claim: by Chernoff bound, we have that
\begin{align*}
    \Pr\left[f(\bw) < \frac{3-\varepsilon}{2} \cdot n\right] \leq e^{-\frac{\varepsilon^2 n}{12}},
\end{align*}
since all boxes realizations are independent Bernoulli random variables and $\E[f(\bw)] = \frac{3}{2}n$. Let us now consider any algorithm, and define the following random variables: $A$ is the number of times the second box in each pair realized in a $2$, while $B_i$ indicates whether or not the algorithm has selected the larger of the two boxes in the $i$-th pair, and we denote $B := \sum_{i \in [n]} B_i$. It is easy to verify that $A$ is binomial, and also $B$ is since $B_i$'s are independent Bernoulli random variables, even though $A$ and $B$ are dependent. We also observe that, by law of total probability,
\begin{align*}
    \Pr\left[a(\bw) >  \frac{2+\varepsilon}{2} \cdot n\right] &= \underbrace{\Pr\left[a(\bw) > \frac{2+\varepsilon}{2} \cdot n \given f(\bw) < \frac{3-\varepsilon}{2} \cdot n\right]}_{\leq 1} \cdot \Pr\left[f(\bw) < \frac{3-\varepsilon}{2} \cdot n\right] \\
    &+ \Pr\left[a(\bw) > \frac{2+\varepsilon}{2} \cdot n \given f(\bw) \geq \frac{3-\varepsilon}{2} \cdot n\right] \cdot \underbrace{\Pr\left[f(\bw) \geq \frac{3-\varepsilon}{2} \cdot n\right]}_{\leq 1}\\
    &\leq e^{-\frac{\varepsilon^2 n}{12}} + \Pr\left[A + B > \frac{2+\varepsilon}{2} \cdot n \given A \geq \frac{1-\varepsilon}{2} \cdot n\right].
\end{align*}
Hereby, the last inequality derives from the fact that $f(\bw) = n + A$, while $a(\bw) = n + A - (n - B) = A + B$. To bound the last term, we can write
\begin{align*}
    \Pr\left[A + B > \frac{2+\varepsilon}{2} \cdot n \given A \geq \frac{1-\varepsilon}{2} \cdot n\right] &\leq \Pr\left[A + B > \frac{2+\varepsilon}{2} \cdot n\right] + \Pr\left[A < \frac{1-\varepsilon}{2} \cdot n\right]\\
    &\leq \Pr\left[B > \frac{2+\varepsilon}{4} \cdot n\right] + \Pr\left[A > \frac{2+\varepsilon}{4} \cdot n\right] + \Pr\left[A < \frac{1-\varepsilon}{2} \cdot n\right]\\
    &\leq 2e^{-\frac{\varepsilon^2 n}{16 + 4\varepsilon}} + e^{-\frac{\varepsilon^2 n}{12}} \\
    &\leq 3e^{-\frac{\varepsilon^2 n}{12}},
\end{align*}
where the second step follows from the union bound. All in all, we can bound the probability that any algorithm obtains a large overall value by $\Pr\left[a(\bw) >  \frac{2+\varepsilon}{2} \cdot n\right] \leq 4e^{-\frac{\varepsilon^2 n}{12}}$. Thus,
\begin{align*}
    \Pr[a(\bw) < \left(\frac{2}{3}+\varepsilon\right)\cdot f(\bw)] &\geq 1 - \Pr\left[f(\bw) < \frac{3-\varepsilon}{2} \cdot n \vee a(\bw) > \frac{2+\varepsilon}{2} \cdot n \right]\\
    &\geq 1 - e^{-\frac{\varepsilon^2 n}{12}} - 4e^{-\frac{\varepsilon^2 n}{12}}\\
    &\geq 1 - 5e^{-\frac{\varepsilon^2 n}{12}},
\end{align*}
where the second step follows again from the union bound.
\end{proof}

\begin{claim}\label{clm:negl-pbm}
There exists a family of feasibility constraints $\cF$ such that both $\RoE(\cF)$ and $\EoR(\cF)$ are constants, but $\PbM(\cF) \leq 2^{-\frac{|E|}{2}}$, i.e., the probability of selecting the maximum is exponentially small in the cardinality of the ground set.
\end{claim}
\begin{proof}
Let us take \Cref{ex:ex2}. For the first part, as seen in Claim \ref{clm:two-thirds}, the algorithm that selects either all the first elements or all the second elements of the pairs guarantees an $\RoE$ and an $\EoR$ of at least $1/2$. Moreover, by Jensen's Inequality, we have
\begin{align*}
    \underbrace{\E\left[\frac{a(\bw)}{f(\bw)}\right]}_{\leq \EoR(\cF,D)} &\geq \overbrace{\frac{ \E\left[a(\bw)\right]}{\E\left[f(\bw)\right]}}^{\leq \RoE(\cF,D)} =  \frac{n}{\E\left[\sum_{i \in [n]}\max\{w_{1,i}, w_{2,i}\}\right]} = \frac{n}{\frac{3}{2}n} = \frac{2}{3}.
\end{align*}
For the second part of the claim, let us observe that the probability of selecting the maximum is the probability that, for each pair of boxes $i$, any algorithm $\ALG$ chooses the larger realized value of the two boxes, $w_{1,i}, w_{2,i}$. Denoting by $w_{\ALG, i}$ the weight of the box chosen by the algorithm for pair $i$, we have
\begin{align*}
    \PbM(\cF) = \Pr[\forall i \in [n], w_{\ALG, i} = \max\{w_{1,i}, w_{2,i}\}] \leq 2^{-n},
\end{align*}
where the last equality follows from the following two facts: first, no online algorithm can choose the maximum of each pair with probability greater than $\frac{1}{2}$ by construction. Moreover, all realizations are independent, and thus we get the claimed upper bound since $n = \frac{|E|}{2}$.
\end{proof}

\section{Omitted Proofs from Section \ref{sec:model}}

\begin{proof}[Proof of \Cref{prop:bounding}]
Let $\bu, \bv \in \mathbb{R}_{\geq 0}^{|E|}$ be nonnegative weight vectors and $\OPT(\bu), \OPT(\bv)$ be, respectively, their optimal feasible subset. 
\begin{enumerate}
    \item \textbf{Lipschitz.} We assume, without loss of generality, that $f(\bu) \geq f(\bv)$.  Then 
    \begin{eqnarray*}
         \mid f(\bu) - f(\bv)\mid &= & f(\bu) - f(\bv) \\ & =& \sum_{e \in \OPT(\bu)}{u_e} - \sum_{e \in \OPT(\bv)}{v_e}
         \\ & \leq& \sum_{e \in \OPT(\bu)}{u_e} - \sum_{e \in \OPT(\bu)}{v_e}
           \\ & \leq& \sum_{e \in \OPT(\bu)}{\mid u_e-v_e\mid } 
         \\ & \leq & \sum_{e \in E}{\abs{u_e - v_e}} = \norm{\bu - \bv}_1,
    \end{eqnarray*}
    where the first inequality is by \Cref{def:f}.
    We conclude that $f$ is 1-Lipschitz.
    \item \textbf{Monotonicity.} Let $\bu,\bv$ be two vectors such that $\bv \leq \bu$ component-wise. It holds that
    \begin{align*}
       f(\bv)=\bv(\OPT(\bv)) \leq  \bu(\OPT(\bv)) \leq    \bu(\OPT(\bu)) =f(\bu),
    \end{align*}
    where the first inequality is by our assumption that $\bu\geq \bv$, and the second inequality is by the definition of $\OPT$. 
    \item \textbf{Self-boundness.} For every element $e\in E$, we define $\bw^{(e)} := (w_1, \dots, w_{e-1}, w_{e+1}, \dots, w_{|E|})$ to be the weight vector $\bw$ deprived of $e^{\text{th}}$ coordinate and $f_e: \mathbb{R}_{\geq 0}^{|E|-1} \rightarrow \mathbb{R}$ to be the function $f$ that takes as input a weight vector $\bw^{(e)}$ and apply $f$ to the vector $(\bw^{(e)},w_e=0)$.  Thus $f/\tau$ in the range of $[0,\tau]^E$ is self-bounding since: 
    \begin{itemize}
        \item The first condition of \Cref{def:bounding} that $0 \leq \frac{1}{\tau}f(\bw) - \frac{1}{\tau}f_e(\bw^{(e)}) \leq 1$ is satisfied by monotonicity, the $1$-Lipschitzness, and the fact that $w_e$ is restricted to the range of $[0,\tau]$.
        \item The second condition of \Cref{def:bounding} holds, since
    \begin{align*}
        \sum_{e \in E}\left(\frac{1}{\tau}f(\bw) - \frac{1}{\tau}f_e(\bw^{(e)})\right) &= \sum_{e \in \OPT(\bw)}\left(\frac{1}{\tau}f(\bw) - \frac{1}{\tau}f_e(\bw^{(e)})\right) \leq \frac{1}{\tau} \cdot \sum_{e \in \OPT(\bw)}{w_e} = \frac{1}{\tau} \cdot f(\bw),
    \end{align*}
    where the first equality derives from the fact that if $e \notin \OPT(\bw)$, then $f(\bw) - f_e(\bw^{(e)}) = 0$, the inequality is by the $1$-Lipschitzness, and the last equality is by the definition of $f$.
    \end{itemize}  
\end{enumerate}
This concludes the proof.
\end{proof}

\section{Auxiliary Notions from Section \ref{sec:relation}}\label{app:omitted}

In \Cref{sec:relation}, we have used the following definitions and claims, as well as the properties of self-bounding functions (see \Cref{def:bounding}).

\begin{claim}[Jensen's Inequality]\label{jensen-orig} Let $u$ be a random variable and $g$ be any function, then
\begin{align*}
    g\left(\E\left[u\right]\right) &\leq \E\left[g(u)\right], \text{ if } g \text{ is convex.}\\
    g\left(\E\left[u\right]\right) &\geq \E\left[g(u)\right], \text{ if } g \text{ is concave.}
\end{align*}
\end{claim}

\begin{claim}\label{jensen}
Let $a > 0$ be a constant, and $u, v > 0$ two independent random variables, then we have
\begin{align*}
    \frac{a}{a+\E\left[u\right]} &\leq \E\left[\frac{a}{a+u}\right],\\
    \frac{\E\left[u\right]}{a+\E\left[v\right]} &\leq \E\left[\frac{u}{a+v}\right].
\end{align*}
\end{claim}
\begin{proof}
The first inequality is a direct consequence of Jensen's Inequality (Claim \ref{jensen-orig}) with function $g(X) = a/(a+X)$ which is convex. The second inequality above holds because by independence of $x,y$, we have
\begin{align*}
    \E\left[\frac{X}{a+Y}\right] = \E\left[X\right] \cdot \E\left[\frac{1}{a+Y}\right] \geq \frac{\E\left[X\right]}{a+\E\left[Y\right]}.
\end{align*}
where the last step follows again from Jensen's Inequality. 
\end{proof}

\begin{claim}\label{avg-pb}
For any collection of $p_1, \dots, p_n \in (0, 1)$, let $p := \left(\prod_{i \in [n]}{p_i}\right)^{1/n}$. We have that \begin{align*}
    n\cdot\frac{1-p}{p} \leq \sum_{i \in [n]}{\frac{1-p_i}{p_i}}
\end{align*}
\end{claim}

\begin{proof}
Let us consider the base case for $p_1, p_2 \in (0,1)$: 
\begin{align*}
    \frac{2(1-\sqrt{p_1p_2})}{\sqrt{p_1p_2}} \leq \frac{1 - p_1}{p_1} + \frac{1 - p_2}{p_2}
\end{align*}

We now obtain the following,

\begin{align*}
\frac{n(1-p)}{p} &= n \cdot \frac{1 - e^{\log\left(\prod_{i \in [n]}{p_i}\right)^{1/n}}}{e^{\log\left(\prod_{i \in [n]}{p_i}\right)^{1/n}}} = n \cdot \frac{1 - e^{\frac{1}{n}\sum_{i \in [n]}{\log p_i}}}{e^{\frac{1}{n}\sum_{i \in [n]}{\log p_i}}} \\
& \leq n \cdot \frac{1 - 1/n\sum_{i\in[n]}{p_i}}{1/n\sum_{i\in[n]}{p_i}}\\
& \leq \sum_{i \in [n]}{\frac{1-p_i}{p_i}}
\end{align*}

Hereby, the first inequality is a consequence of the AM-GM inequality. 
\end{proof}

\begin{claim}\label{opts}
For every realization of $\bw$, it holds that $f(\bw) \leq f(\overline{\bw}) + \sum_{e \in E}{w_e \cdot \bone\left[w_e > \tau\right]}$.
\end{claim}
\begin{proof}
Let $S = \{e \mid w_e \leq \tau$\}. We have that
\begin{eqnarray*}
    f(\bw)  &= &\sum_{e \in \OPT(\bw) \cap S}w_e  + \sum_{e \in \OPT(\bw) \setminus S}w_e  \\
    & = &  \sum_{e \in \OPT(\bw) \cap S}{\overline{w}_e}  + \sum_{e \in \OPT(\bw) \setminus S}{w_e } \\
    & \leq &  f(\overline{\bw}) + \sum_{e \in E \setminus S}{w_e }  = f(\overline{\bw}) + \sum_{e \in E}{w_e \cdot \bone\left[w_e > \tau\right]},
\end{eqnarray*}
where the second equality is since $\overline{w}_e=w_e$ for every element $e \in S$, and the inequality is since $\OPT(\bw) \setminus S$ is a feasible set, and $\OPT(\bw) \setminus S \subseteq E\setminus S$.
\end{proof}

\begin{claim}[Adapted from Theorem 1 in \citep{esfandiari}]\label{gamma-lemma}
For any set of distributions $\{D_e\}_{e \in E}$,
\begin{align*}
    \Pr\left[\cE_0\right] &= \gamma,\\
    \Pr\left[\cE_1\right] &\geq \gamma \log\left(\frac{1}{\gamma}\right).
\end{align*}
\end{claim}
\begin{proof}
By the way we have defined $\tau$, we have $\Pr\left[\tau > \max_{e \in E}{w_e}\right] = \gamma$. This is equivalent to saying
\begin{align*}
    \Pr\left[\forall \ e \in E: \ w_e \leq \tau\right] = \gamma
\end{align*}
We immediately notice that the probability of the complement of above is the same as the probability of at least one $w_e$ with value at least $\tau$,
\begin{align*}
    \Pr\left[\tau < \max_{e \in E}{w_e}\right] = \Pr\left[\exists \ e \in E: \ w_e > \tau\right]
\end{align*}

We also know that the RHS can be written as follows by union bound of disjoint events: letting $p_e := \Pr\left[w_e < \tau\right]$,
\begin{align*}
    \Pr\left[\exists \ e \in E: \ w_e > \tau\right] &\geq \Pr\left[\exists ! \ e \in E: \ w_e > \tau\right] \\
    &= \sum_{e \in E}{\left((1-p_e)\cdot \prod_{e^\prime \neq e}{p_{e^\prime}}\right)} \\
    &= \prod_{e \in E}{p_e} \cdot \sum_{e \in E}{\frac{1-p_e}{p_e}}
\end{align*}

Hereby, the first equality follows from the fact that events of the form ``$\exists ! \ e \in E: \ w_e > T$" are disjoint, and the union bound holds with equality. Let us note that the threshold is larger than the maximum over all realizations if and only if all $w_e$'s realize in a value less than the threshold $\tau$, which happens by definition with probability $\gamma$, which means that $\prod_{e \in E}{p_e} = \gamma$. By Claim \ref{avg-pb} (below), we let $p := \left(\prod_{e \in E}{p_e}\right)^{1/|E|}$ and get
\begin{align*}
     \sum_{e \in E}{\frac{1-p_e}{p_e}} &\geq |E| \cdot \frac{(1-p)}{p} \\
     &\geq |E| \cdot \frac{1-\gamma^{1/|E|}}{\gamma^{1/|E|}} \geq \log\left(\frac{1}{\gamma}\right)
\end{align*}

The second inequality follows from noticing again that $p = \gamma^{1/|E|}$, while the third by observing that the LHS of the same inequality is non-increasing in $|E|$ and the smallest value is attained for $|E| \rightarrow \infty$. All in all, we get that 
\begin{align*}
    \Pr\left[\exists ! \ e \in E: \ w_e > \tau\right] \geq \gamma \cdot \log\left(\frac{1}{\gamma}\right)
\end{align*}
and the statement follows.
\end{proof}




\section{Extensions of Our Results}

\subsection{Extension to Single-Sample}\label{app:ss}
In this appendix, we discuss an extension of our results to the case of single sample settings. In our results and definitions, we assumed that the algorithm has  \emph{full knowledge} of the distributions from where the elements' weights are drawn. This is used to distinguish between running two different subroutines: the first selects a single element, while the second run the combinatorial algorithm (see also \Cref{general,alg:eor-roe}). On the other hand, as mentioned in \Cref{relwork}, there has been a significant effort towards designing prophet inequalities with provable guarantees when the decision-maker has only access to some samples from each distribution. 
We explain now how our result can be extended to single sample settings.
To do so, we first define equivalent definitions to those of \Cref{sec:model}. 

To measure our performance, given a downward-closed family $\cF$, an algorithm $\ALG$, and a product distribution $D$, we define
\begin{align*}
    \SEoR(\cF, \ALG, D) := \E\left[\frac{\bw(\ALG(\bw))}{f(\bw)}\right],
\end{align*}
where the expectation runs over the stochastic generation of the input, as well as the (possible) randomness of the algorithm.
Similarly, we define 
\begin{equation}
    \SEoR(\cF,\ALG) := \inf_{D} \SEoR(\cF, \ALG, D), \label{eq:SEORD}
\end{equation}
and
\begin{equation}
    \SEoR(\cF):= \sup_{\ALG}  \SEoR(\cF,\ALG). \label{eq:SEORF} 
\end{equation}

We note that since we consider sample-based algorithms, the algorithm is oblivious to the distributions, but may depend on the sampled values, and the online values observed so far.

We will compare our results to the standard objective of maximizing the ratio of expectations between the algorithm and the offline optimum. Accordingly, we denote
\begin{align*}
    \SRoE(\cF, \ALG, D) := \frac{\E\left[\bw(\ALG(\bw))\right]}
    {\E\left[{f(\bw)}\right]}.
\end{align*}

Analogously to Equations~\eqref{eq:SEORD},~and \eqref{eq:SEORF}, we define $\SRoE(\cF,\ALG)$ and  $\SRoE(\cF)$.

\begin{corollary}[Single-sample simulation]\label{cor:ss}
For every downward-closed family of feasibility constraints $\cF$, it holds that
\begin{align}
\emph{\SEoR}(\cF) \geq \frac{\emph{\SRoE}(\cF)}{144}.
\end{align}
\end{corollary}

\begin{proof}
Let us first observe that we cannot compute the expected optimum of an instance and perform the case distinction, but we can still flip a fair coin and run either of the two following subroutines. With probability $1/2$ (the coin lands heads) we select the first element that exceeds $\hat \tau := \max_{e \in E} s_e$, and with probability $1/2$ (the coin lands tails), we run the $\SALG$ subroutine, an algorithm that achieves $\SRoE = \alpha$ only using samples, and which exists by the assumption of the claim. Hereby, we denote by $S = \{ s_e \mid e \in E \}$ the sample set, and by $O = \{ o_e \mid e \in E \}$ the online set. Moreover, let $R = S \cup O$ and $r_1 > r_2 > \dots > r_{2|E|}$ denote the weights in $R$ sorted in decreasing order.
Although the values $W,\tau$ (see analysis of \Cref{alg-analysis}), are not known to the algorithm, the analysis can still be partitioned as in the analysis of \Cref{alg-analysis}.

If $W \leq c \cdot \tau$ and the coin lands heads, then our algorithm selects the first online value $o_e \geq \hat \tau$. Let us only count cases where the following conditions occur jointly: (1) exactly one element in $S$ and exactly one element in $O$ exceed $\tau$, (2) the element with the largest weight is in the online set. It is not difficult to see that this happens with probability at least $\frac{1}{2} \cdot \gamma^2\log^2(1/\gamma)$. In this case, the approximation, following the analysis of \Cref{newcase1}, is $1/(c+1)$. From this first case, we obtain
\begin{align*}
    \SEoR(\cF) \geq \frac{\gamma^2\log^2(1/\gamma)}{4(c+1)}.
\end{align*}

If $W > c \cdot \tau$ and the coin lands tails, we run algorithm $\SALG$ that, given a single sample from each distribution in the input, satisfies $\E\left[\SALG(\bw)  \mid \cE_0 \right]  \geq \alpha \cdot  \E\left[f(\bw) \mid  \cE_0\right]$. As before, we only count cases where no sample and no online element exceeds threshold $\tau$, which happens with probability $\gamma^2$. Following the analysis of \Cref{newcase2}, from this case we have that
\begin{equation*}
        \SEoR(\cF) \geq \frac{\gamma^2}{2\delta}\frac{k-\delta}{k}\alpha ,
\end{equation*}
which means that we lose a multiplicative factor of $\gamma/2$ compared to the full information case. All in all, the described algorithm guarantees
\begin{align*}
    \SEoR(\cF) \geq \min\left\{\frac{\log^2 2}{16 \cdot (\frac{8}{3} \log \frac{3}{\alpha} +1)}, \frac{\alpha}{72}\right\} \geq \frac{\alpha}{144}, 
\end{align*}
where we have chosen the parameters as in \Cref{alg-analysis}. This concludes the proof.
\end{proof}

Note that, for the other direction, running each of the subroutines used in \Cref{alg:eor-roe} with probability $1/2$, we get $\SRoE(\cF, D) \geq \SEoR(\cF, D)/136$.

\subsection{Extension to XOS Functions}\label{app:xos} 
In this appendix, we describe in detail the extension of results contained in \cref{sec:relation,sec:reduction2} to a setting where the decision maker's objective function is XOS, rather than just additive.

\begin{definition}\label{def:xos}
A function $g: 2^{|E|} \to \mathbb{R}_{\geq 0}$ is XOS if there exist $\ell$ vectors $b_1, \dots, b_{\ell} \in \mathbb{R}_{\geq 0}^{|E|}$, such that $g(S) = \max_{i \in [\ell]} \sum_{j\in S}(b_i)_j $.
\end{definition}
We extend \Cref{def:xos} to vectors of weights in the following way:
\begin{definition}
A function $g: \mathbb{R}_{\geq 0}^{|E|} \to \mathbb{R}_{\geq 0}$ is extended-XOS if there exist $\ell$ vectors $b_1, \dots, b_{\ell} \in \mathbb{R}_{\geq 0}^{|E|}$, such that $g(\bw) = \max_{i \in [\ell]} \ev{b_i, \bw}$.
\end{definition}
In Claim \ref{cl:xos}, we show that the projection of an extended-XOS function to $\{0,1\}^{|E|}$ is an XOS function, and that the projections of all extended-XOS functions to $\{0,1\}^{|E|}$, are all XOS functions.

We next extend \Cref{def:f} to extended-XOS functions:

\begin{definition}\label{def:f-xos}
Let $\bw \in \mathbb{R}_{\geq 0}^{|E|}$ be a nonnegative weight vector. We define $\OPT: \mathbb{R}_{\geq 0}^{|E|} \to \cF$ to be the function mapping a vector of weights to a maximum-weight set in family $\cF$. Namely,
\begin{align*}
    \OPT(\bw) = \arg\max_{S \in \cF} \max_{i \in [\ell]} \ev{b_i, \bw_S},
\end{align*}
where $\bw_S \in \mathbb{R}^{|E|}$ is the vector of elements weights in $S$ (and $0$ for elements not in $S$), while $b_i \in \mathbb{R}^{|E|}$ is a vector of nonnegative coefficients. Moreover, we let $f_{\cF}(\bw) =\max_{S \in \cF} \max_{i \in [\ell]} \ev{b_i, \bw_S}$.
\end{definition}
For simplicity, when clear from context, we denote $f_{\cF}$ by $f$.
\begin{claim}\label{cl:xos}
$ f(\bw)$ can be expressed as $f(\bw) = \max_{\substack{i \in [\ell]\\S \in \cF}}\ev{a^S_i, \bw_S}$ and is extended-XOS. Moreover, the function resulting from projecting $f$ onto $\{0,1\}^{|E|}$ (i.e. $f: \{0,1\}^{|E|} \rightarrow \mathbb{R}$ is now a set function) is XOS and describes all XOS functions.
\end{claim}
\begin{proof}
For the first part of the claim, consider a set $T$. Then,
\begin{align}
    f(T) &= \max_{S \in \cF}\max_{i \in [\ell]} \sum_{j \in S \cap T} b_{ij}\nonumber\\
    &= \max_{S \in \cF}\max_{i \in [\ell]} \sum_{j \in T} a^S_{ij}, \label{eq:xos-project}
\end{align}
where the second equality follows by setting $a^S_{ij} := b_{ij} \cdot \bone_{j \in S}$. Hence, $f(\bw) = \max_{\substack{i \in [\ell]\\S \in \cF}}\ev{a^S_i, \bw_S}$ and it is extended-XOS with coefficient vectors $a^S_i$, as per \Cref{def:f-xos}.

The second part of the claim follows immediately by Equation \eqref{eq:xos-project} since this is maximum over additive functions.

The third part of the claim follows since for every XOS function $f$, which is defined with additive functions $a_1,\ldots,a_\ell$, the projection of the extended-XOS function defined by the same additive functions, is $f$.
\end{proof}

We have the following properties for $f$.

\begin{proposition}[Properties of XOS $f_{\cF}$]\label{prop:bounding-xos}
For every downward-closed family of sets $\cF$, the function $f$ $(=f_\cF)$ satisfies the following properties:
\begin{enumerate}
    \item $f$ is ($\max_{i,j}a_{ij}$)-Lipschitz.
    \item $f$ is monotone, i.e., if $\bu \geq \bv$ point-wise, then $f(\bu)\geq f(\bv)$.
    \item For every $\tau > 0$, the function  $\frac{f}{\tau \cdot \max_{i,j}{a_{ij}}}$ restricted to the domain $[0,\tau]^{|E|}$, is self-bounding.
\end{enumerate}
\end{proposition}
\begin{proof}
Let $\bu, \bv \in \mathbb{R}_{\geq 0}^{|E|}$ be nonnegative weight vectors and $\OPT(\bu), \OPT(\bv)$ be, respectively, their optimal feasible subset.
\begin{enumerate}
    \item \textbf{Lipschitz.} We assume, without loss of generality, that $f(\bu) \geq f(\bv)$.  Then, letting $i(\bu), i(\bv)$ be respectively the indices of the coefficients maximizing $f(\bv), f(\bu)$, we get
    \begin{eqnarray*}
         \mid f(\bu) - f(\bv)\mid &= & f(\bu) - f(\bv) \\
         & =& \max_{i \in [\ell]} \sum_{j \in \OPT(\bu)} a_{ij}u_j - \max_{i \in [\ell]} \sum_{j \in \OPT(\bv)} a_{ij}v_j
         \\
         & =& \sum_{j \in \OPT(\bu)} a_{i(\bu)j}u_j - \sum_{j \in \OPT(\bv)} a_{i(\bv)j}v_j
         \\
         & \leq& \sum_{j \in \OPT(\bu)} a_{i(\bu)j}u_j - a_{i(\bv)j}v_j
         \\ & \leq& \max_{i,j}a_{ij} \cdot \sum_{j \in \OPT(\bu)}{\mid u_j-v_j\mid } 
         \\ & \leq & \max_{i,j}a_{ij} \cdot \sum_{j \in E}{\abs{u_j - v_j}} = \max_{i,j}a_{ij} \cdot \norm{\bu - \bv}_1,
    \end{eqnarray*}
    \item \textbf{Monotonicity.} Let $\bu,\bv$ be two vectors such that $\bv \leq \bu$ component-wise. Then, we know that $\ev{a_i, \bv} \leq \ev{a_i, \bu}$ for all $i \in [\ell]$, and thus letting $i(\bu), i(\bv)$ be respectively the indices of the coefficients maximizing $f(\bv), f(\bu)$, we get
    \begin{align*}
       f(\bv) = \ev{a_{i(\bv)}, \bv} \leq \ev{a_{i(\bv)}, \bu} \leq \ev{a_{i(\bu)}, \bu} = f(\bu).
    \end{align*}
    \item \textbf{Self-boundness.} For every element $e\in E$, we define $f_e$ to be the function $f$ where the weight of element $e$ is $0$. I.e., $f_e(\bw^{(e)}) := f(w_1, \dots, w_{e-1},0, w_{e+1}, \dots, w_{|E|})$. Thus $\frac{f}{\tau \cdot \max_{i,j}{a_{ij}}}$ in the range of $[0,\tau]^{|E|}$ is self-bounding since: 
    \begin{itemize}
        \item The first condition of \cref{def:bounding} that $0 \leq \frac{1}{\tau\cdot \max_{i,j}{a_{ij}}}f(\bw) - \frac{1}{\tau\cdot \max_{i,j}{a_{ij}}}f_e(\bw^{(e)}) \leq 1$ is satisfied by monotonicity, the $\max_{i,j}a_{ij}$-Lipschitzness, and that fact that $w_e$ is restricted to the range of $[0,\tau]$.
        \item We multiply the second condition of \cref{def:bounding} by factor $\tau \cdot \max_{i,j}{a_{ij}}$ and get
    \begin{align*}
        \sum_{e \in E}\left(f(\bw) - f_e(\bw^{(e)})\right) &= \sum_{e \in \OPT(\bw)} \underbrace{f(\bw)}_{ = \ev{a_{i^*}, \bw}} - \underbrace{f_e(\bw^{(e)})}_{ \geq \ev{a_{i^*}, \bw^{(e)}}} \\
        &\leq \sum_{e \in \OPT(\bw)}{\ev{a_{i^*}, \bw - \bw^{(e)}}} \\
        &= \sum_{e \in \OPT(\bw)}{a_{i^*e} \cdot w_e} = f(\bw),
    \end{align*}
    where the first equality derives from the fact that if $e \notin \OPT(\bw)$, then $f(\bw) - f_e(\bw^{(e)}) = 0$, and the first inequality derives by setting $a_{i^*}$ to be the maximizing coefficient for $f$.
    \end{itemize}  
\end{enumerate}
This concludes the proof.
\end{proof}

\begin{remark}
Before proceeding, we note that \Cref{prop:bounding-xos} generalizes similar claims in previous literature. Namely, \citet{Vondrak10} proves self-boundness for set functions which are XOS or submodular. The key difference between our claim and those in \cite{Vondrak10} is that our function $f$ is not a combinatorial function, in that $f$ receives a vector of arbitrary weights instead of a set. Therefore, Lemma 2.2 in \cite{Vondrak10} follows by projecting \Cref{prop:bounding-xos} onto the hypercube.
Similarly, \Cref{prop:bounding-xos} also generalizes Lemma 2.4 in \citet{BlumCHPPV17} where they prove the self-bounding property for the maximum matching cardinality function. 
\Cref{prop:bounding-xos} is again a generalization in the following senses: first, the constraint is not necessarily a matching, but can be a general packing one. Second, $f$ can be an arbitrary extended-XOS function, and not just the cardinality function.
\end{remark}

Below, we illustrate generalizations of feasibility-based and instance-based reductions under extended XOS functions.

\medskip

\noindent \textbf{Feasibility-based reduction with extended XOS functions.} We redefine threshold $\tau$ as follows for reasons that will be explained later,
\begin{align*}
    \tau : \Pr\left[\exists j \in E: w_j > \frac{\tau}{\max_{i \in [\ell]}a_{ij}}\right] = \gamma.
\end{align*}

Given the above, let us run \Cref{general}, performing the case distinction with the redefined threshold $\tau$. In the ``catch the superstar'' case (\Cref{newcase1}), the algorithm selects the first element $j$ for which $w_j > \frac{\tau}{\max_{i \in [\ell]}a_{ij}}$. Hence, with probability $\gamma \cdot \log(1/\gamma)$ we catch a unique element with value $w_j \cdot \max_{i,j}{a_{ij}} > \tau$. Since, in this case $W \leq c \cdot \tau$, we get a $\frac{\gamma \cdot \log(1/\gamma)}{c + 1}$ expected ratio.

For the ``run the combinatorial algorithm'' case (\Cref{newcase2}), we have that, since $\frac{f}{\tau \cdot \max_{i,j}{a_{ij}}}$ is self-bounding in $[0, \tau]$ (\Cref{prop:bounding-xos}), we observe that in order to maintain the same concentration, we need $c \geq \frac{4 + 2\delta}{3(\delta - 1)^2} \log \frac{k}{\alpha} \cdot \max_{i,j}{a_{ij}}$. Thus, for the extended-XOS function $f$, we will achieve a similar performance to \Cref{alg-analysis} but losing an additional factor of $\max_{i,j}{a_{ij}}$: namely,
\begin{align*}
    \EoR(\cF) \geq \frac{\RoE(\cF)}{12 \cdot \max_{i,j}{a_{ij}}}.
\end{align*}

\noindent \textbf{Instance-based reduction with extended XOS functions.} An almost identical analysis to the above follows in the case of \Cref{alg:eor-roe}. Namely, let us redefine $A$ to be
\begin{align*}
    A := \E\left[\max_{i,j} a_{ij} \cdot w_j\right].
\end{align*}
With this at hand, \Cref{alg-analysis-rev-new} now corresponds to the following statement, again with $f_{\cF}$ being an extended-XOS function:
\begin{align*}
    \RoE(\cF, D) \geq \frac{\EoR(\cF, D)}{68 \cdot \max_{i,j}{a_{ij}}}.
\end{align*}

\section{Alternatives to the{} \EoR} \label{app:alternative}

In this appendix, we delve into the discussion regarding alternative measures to the \EoR, expanding upon what introduced in \Cref{sec:intro}. In particular, the following extension of \PbM, named for simplicity $\PbM_p$, is studied in \cite{HoeferK17,SotoTV21,BahraniBSW21}. The algorithm's goal is to select each element in the ex-post optimal set with probability at least $p$, for largest possible  value of $p$. Formally, we have that
\begin{align*}
    \PbM_p(\cF, D, \ALG) := \min_{e \in E} \Pr[e \in \ALG(\bw) \mid e \in \OPT(\bw)],
\end{align*}
and 
\begin{align*}
    \PbM_p(\cF) := \inf_D \sup_{\ALG} \PbM_p(\cF, D, \ALG).
\end{align*}
Note that for single-choice settings, the above measure reduces to $\PbM$. In \citet[Definition~3]{BahraniBSW21},  this notion is referred to as \emph{probability-competitive} algorithms. It is a major open question 
whether for general matroids in secretary settings, there exists a constant probability-competitive algorithm.

We now show the shortcoming of this measure for downward-closed prophet settings, via the following example. In particular, we present a feasibility constraint and an instance  in which it is possible to achieve a good approximation, but no algorithm can guarantee a performance of better than $O(1/n)$ according to  $\PbM_p$. 

\begin{example}\label{4}
Consider a setting with $n$ pairs of boxes, such that, for each pair $i$, one box has weight $w_{1,i} = 1$ deterministically, and the second has weight $w_{2,i} \sim \cU[1,2]$. All the deterministic boxes arrive first and the rest arrive later. The (downward-closed) feasibility constraint is that only elements from at most one pair can be selected. 
\begin{figure}[H]
    \centering
    \scalebox{0.6}{\tikzset{every picture/.style={line width=0.75pt}} 

\begin{tikzpicture}[x=0.75pt,y=0.75pt,yscale=-1,xscale=1]

\draw   (20,89.6) -- (32.6,77) -- (117,77) -- (117,106.4) -- (104.4,119) -- (20,119) -- cycle ; \draw   (117,77) -- (104.4,89.6) -- (20,89.6) ; \draw   (104.4,89.6) -- (104.4,119) ;
\draw   (104.4,89.6) -- (117,77) -- (201.4,77) -- (201.4,106.4) -- (188.8,119) -- (104.4,119) -- cycle ; \draw   (201.4,77) -- (188.8,89.6) -- (104.4,89.6) ; \draw   (188.8,89.6) -- (188.8,119) ;
\draw   (448,86.6) -- (460.6,74) -- (545,74) -- (545,103.4) -- (532.4,116) -- (448,116) -- cycle ; \draw   (545,74) -- (532.4,86.6) -- (448,86.6) ; \draw   (532.4,86.6) -- (532.4,116) ;
\draw   (532.4,86.6) -- (545,74) -- (629.4,74) -- (629.4,103.4) -- (616.8,116) -- (532.4,116) -- cycle ; \draw   (629.4,74) -- (616.8,86.6) -- (532.4,86.6) ; \draw   (616.8,86.6) -- (616.8,116) ;

\draw (289,89.4) node [anchor=north west][inner sep=0.75pt]    {$\dotsc \dotsc \dotsc $};
\draw (53,94.4) node [anchor=north west][inner sep=0.75pt]  [font=\Large]  {$D_{1}$};
\draw (139,95.4) node [anchor=north west][inner sep=0.75pt]  [font=\Large]  {$D_{2}$};
\draw (481,91.4) node [anchor=north west][inner sep=0.75pt]  [font=\Large]  {$D_{1}$};
\draw (567,92.4) node [anchor=north west][inner sep=0.75pt]  [font=\Large]  {$D_{2}$};

\end{tikzpicture}}
    \caption{$n$ pairs of boxes: for each pair, $w_{1,i} = 1$, $w_{2,i} \sim D_2 = \cU[1,2]$.}
    \label{fig:ex4}
\end{figure}
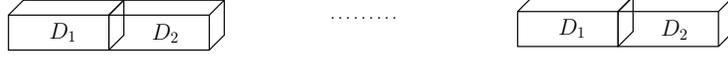
According to the measure of $\PbM_p$, the optimal set will always be a pair of two boxes, and any algorithm will have to guess what the pair exactly is in order to select each element of the optimal set, which happens with probability at most $1/n$. After the arrival of the first $n$ deterministic boxes, let us consider an element the algorithm has selected with probability at most $1/n$: this element will belong to the optimal set with probability $1/n$. Hence, given that $e \in \OPT(\bw)$ with probability $1/n$, the algorithm will have selected $e$ with probability at most $1/n$, i.e. $\PbM_p(\cF) \leq 1/n$. Since the trivial algorithm that selects a random pair always achieves a $\PbM_p(\cF) = 1/n$ and, as shown, no other algorithm can have better performance, this measure becomes uninformative of the algorithm's quality in combinatorial settings.

It is easy to see that both the $\RoE$ and $\EoR$ are both at least a constant  for this specific $\cF$: let us consider the single-choice algorithm that tries to catch box of maximum weight out of the $2n$ boxes, and then selects the pair corresponding to this box. This algorithm is essentially treating the instance in \Cref{fig:ex4} as if $2n$ boxes arrived singularly. As discussed in \Cref{app:single}, an algorithm aiming to catch the box with maximum weight achieves an $\EoR$ of $1/e$. The algorithm then selects the pair where this box belongs. This means that, with probability $1/e$, any other box can have a weight at most that of the maximum, i.e., the total weight in any other pair can be at most twice as large as the weight the selected pair. Thus, $\EoR(\cF) \geq 1/2e$. Furthermore, by what argued before, we have that, for all $\bw$, the optimal pair weight $f(\bw) \leq 2 \tilde f(\bw)$, where $\tilde f(\bw)$ is the maximum weight out of the $2n$ boxes. Finally, this same algorithm satisfies
\begin{align*}
    \E[a(\bw)] \geq \E[a(\bw) \mid a(\bw) = \tilde f(\bw)] \cdot \Pr[a(\bw) = \tilde f(\bw)] \geq 1/e \cdot \E[\tilde f(\bw)] \geq 1/2e \cdot \E[f(\bw)].
\end{align*}
Hence, $\RoE(\cF) \geq 1/2e$.
\end{example}

\section{Incomparability of the \RoE-\EoR{} framework against Classic Risk Aversion}\label{app:risk}

In this section, we compare risk-neutral algorithms maximizing $\RoE$ or $\EoR$ against risk-averse algorithms maximizing expected utility. The classic definition of risk aversion posits that the expected utility of a risk-averse algorithm is concave in the reward it gets. Namely, if $a(\bw)$ is the reward of the algorithm, then $\E[\text{util}(a(\bw))]$ is concave. 

We show via an example that a risk-neutral algorithm maximizing $\RoE$ or $\EoR$ may yield very poor risk-averse expected utility (in the classic sense of risk aversion). Conversely, a risk-averse algorithm may achieve extremely low $\RoE$ and $\EoR$. Thus, the two risk-aversion frameworks (bi-criteria vs. classic) are incomparable in general.

\begin{example}\label{ex:risk}
    Consider a setting with three boxes. The first box's weight $w_1$  is deterministically $1$, the second is $w_2$  which is $0$ with probability $1-\sqrt{\varepsilon}$ and $\frac{1}{\varepsilon}$ with probability $\sqrt{\varepsilon}$, and the third $w_3$ is $0$ with probability $1-\varepsilon$ and $\frac{2}{\varepsilon^2}$ with probability $\varepsilon$, for $\varepsilon \in (0, 1]$. The feasibility constraint is to select at most one box.
    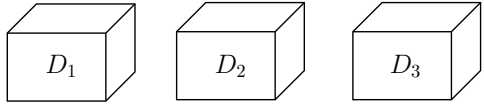
\begin{figure}[H]
        \centering
        \scalebox{0.7}{\tikzset{every picture/.style={line width=0.75pt}} 

\begin{tikzpicture}[x=0.75pt,y=0.75pt,yscale=-1,xscale=1]

\draw   (271,160) -- (292,139) -- (363,139) -- (363,188) -- (342,209) -- (271,209) -- cycle ; \draw   (363,139) -- (342,160) -- (271,160) ; \draw   (342,160) -- (342,209) ;
\draw   (149,161) -- (170,140) -- (241,140) -- (241,189) -- (220,210) -- (149,210) -- cycle ; \draw   (241,140) -- (220,161) -- (149,161) ; \draw   (220,161) -- (220,210) ;
\draw   (398,160) -- (419,139) -- (490,139) -- (490,188) -- (469,209) -- (398,209) -- cycle ; \draw   (490,139) -- (469,160) -- (398,160) ; \draw   (469,160) -- (469,209) ;

\draw (175,175.4) node [anchor=north west][inner sep=0.75pt]    [font=\Large] {$D_{1}$};
\draw (297,175.4) node [anchor=north west][inner sep=0.75pt]    [font=\Large] {$D_{2}$};
\draw (423,175.4) node [anchor=north west][inner sep=0.75pt]    [font=\Large] {$D_{3}$};

\end{tikzpicture}}
        \caption{Three boxes: $w_1 = 1$, $w_2 \sim D_2 = \begin{cases} 0, \text{ w.p. } 1 - \sqrt{\varepsilon} \\ \frac{1}{\varepsilon}, \text{ w.p. } \sqrt{\varepsilon} \end{cases}$, $w_3 \sim D_3 = \begin{cases} 0, \text{ w.p. } 1 - \varepsilon \\ \frac{2}{\varepsilon^2}, \text{ w.p. } \varepsilon \end{cases}$.}
        \label{fig:ex5}
    \end{figure}
    Moreover, suppose that the decision-maker is risk-averse with non-linear utility function $\text{util}(v) = \min\left(v, \frac{2}{\varepsilon}\right)$, where we denote by $v$ the value the algorithm gets, depending on whether it tries to maximize $\RoE$, $\EoR$ or $\E[\text{util}]$. 
    
    In this case, the value the algorithm gets by maximizing the $\RoE$ is $2$, since it would select the last box. Similarly, the value the algorithm obtains by maximizing the $\EoR$ is $1$, since it would select the first box. On the other hand, the algorithm that always selects the second box,  achieves $\E[\text{util}]$ of $\frac{1}{\sqrt{\varepsilon}}$, and therefore both good  $\EoR$ and $\RoE$ algorithms do not necessarily guarantee a good expected risk-averse utility.  

    On the contrary, the algorithm that selects the second box (which achieves a good  $\E[\text{util}]$) does not guarantee a good  $\EoR$ or $\RoE$  since for this algorithm
    \begin{align*}
        \RoE = \frac{\frac{1}{\sqrt{\varepsilon}}}{(1-\varepsilon)(1-\sqrt{\varepsilon}) + \frac{1-\varepsilon}{\sqrt{\varepsilon}} + \frac{2}{\varepsilon}} \leq \frac{\sqrt{\varepsilon}}{2}, \quad \mbox{and} \quad \EoR = (1-\varepsilon)\sqrt{\varepsilon} + \frac{\varepsilon^2\sqrt{\varepsilon}}{2} \leq \sqrt{\varepsilon}.
    \end{align*}
    These obtained values are far smaller compared to the simple guarantees of $\RoE \geq \frac{1}{2}$ and $\EoR \geq \frac{1}{e}$ for single choice settings.

\end{example}

\end{document}